\documentclass[11pt]{article}
\usepackage[utf8]{inputenc}
\usepackage{amsmath,amsfonts,amssymb}
\usepackage{amsthm}
\usepackage{latexsym}
\usepackage[noadjust]{cite}
\usepackage{graphicx}
\usepackage{xcolor}
\usepackage{dirtytalk}
\usepackage{mathtools}
\usepackage[margin=1in]{geometry}
\usepackage{thm-restate}
\usepackage{enumitem}
\usepackage[linesnumbered,ruled]{algorithm2e}
\usepackage[colorlinks=true, allcolors=blue]{hyperref}
\usepackage[nameinlink,capitalise]{cleveref}
\hypersetup{
    citecolor={violet}
}

\newtheorem{theorem}{Theorem}[section]
\newtheorem{corollary}[theorem]{Corollary}
\newtheorem{lemma}[theorem]{Lemma}
\newtheorem{observation}[theorem]{Observation}

\newtheorem{claim}[theorem]{Claim}
\theoremstyle{definition}
\newtheorem{definition}[theorem]{Definition}
\newtheorem{remark}[theorem]{Remark}
\newtheorem{fact}[theorem]{Fact}

\newcommand{\F}{\mathbb{F}}
\newcommand{\N}{\mathbb{N}}
\newcommand{\E}{\mathbb{E}}
\newcommand{\R}{\mathbb{R}}

\newcommand{\rank}{\mathrm{rank}}

\newcommand{\zo}{\{0, 1\}}

\newcommand{\eps}{\epsilon}

\newcommand{\M}{\mathcal{M}}
\newcommand{\I}{\mathcal{I}}
\newcommand{\Ind}{\mathrm{Ind}}

\newcommand{\Hyp}{\mathrm{Hyp}}

\newcommand{\dKL}{d_\mathrm{KL}}

\title{A Near-Optimal Parallel Algorithm for Finding Matroid Bases}
\date{\today}

\author{Sanjeev Khanna\thanks{Courant Institute, New York University, New York, NY, USA.  Supported in part by NSF award CCF-2625203 and AFOSR award FA9550-25-1-0107. Email: {\tt sanjeev.khanna@nyu.edu}.} \and Aaron Putterman\thanks{School of Engineering and Applied Sciences, Harvard University, Cambridge, Massachusetts, USA. Supported in part by the Simons Investigator Awards of Madhu Sudan and Salil Vadhan and AFOSR award FA9550-25-1-0112. Email: \texttt{aputterman@g.harvard.edu}.} \and Junkai Song\thanks{Courant Institute, New York University, New York, NY, USA. Supported in part by NSF award CCF-2625203. Email: \texttt{junkaisong@nyu.edu}.}}

\begin{document}

\pagenumbering{gobble}

\maketitle

\begin{abstract}
We settle the classic question of the parallel complexity of computing a matroid basis, as first posed in the seminal work of Karp, Upfal, and Wigderson (FOCS 1985, JCSS 1988). Our algorithm runs in $O(n^{1/3}\log^{1/3}n)$ rounds, matching the lower bound of KUW up to a $\log^{2/3}(n)$ factor.
\end{abstract}

\clearpage

\pagebreak  

\tableofcontents

\pagebreak

\pagenumbering{arabic}

\section{Introduction}

\paragraph{Background}

In this work, we study the parallel complexity of finding matroid bases. A matroid $\M$ is given by a \emph{ground set} of $n$ elements $E$, along with a set $\I \subseteq 2^E$ called the set of independent sets. This set $\I$ satisfies many convenient properties:
\begin{enumerate}
    \item (Non-empty) $\emptyset \in \I$.
    \item (Downward Closure) For a set $S \subseteq E$, if $S \in \I$, then for any $S' \subseteq S$, $S' \in \I$.
    \item (Extension Property) If $S \in \I$ and $T \in \I$, with $|T| > |S|$, then there exists an element $e \in T \setminus S$ such that $S \cup \{e\} \in \I$.
\end{enumerate}
Although the above definition is abstract, matroids are an elegant generalization of many basic combinatorial objects like cycle-free sets of edges in graphs, and linearly independent sets of vectors in a vector space. 

We will be particularly interested in \emph{bases} of matroids. A basis of a matroid is any set $S \subseteq E$ such that $S \in \I$ and $S$ is \emph{maximal}, in the sense that no element $e \in E - S$ can be added to $S$ while retaining its independence. In graphs, a basis corresponds with a spanning forest, while in a set of vectors, a basis corresponds with the typical notion of a basis (a linearly independent set of spanning vectors). 

Unfortunately, due to the broad spectrum of objects which are captured by matroids, the \emph{number} of matroids on a ground set of $n$ elements grows super-exponentially large; indeed, the number of matroids grows as $2^{2^{\Omega(n)}}$ \cite{BPV15}. For this reason, when designing algorithms for efficient computation on matroids, we use oracle access to the underlying matroid. The most basic form of oracle access (and the oracle that we will use throughout this work) is the notion of an \emph{independence oracle}. Given a matroid $\M$, an independence oracle $\Ind: 2^E \rightarrow \zo$, where for $S \subseteq E$, $\Ind(S) = \mathbf{1}[S \in \I]$.

It is not hard to see that an independence oracle suffices for computing a matroid basis in only $n$ queries; one can simply greedily build a basis by trying to add each element $e \in E$ one at a time. Unfortunately, this algorithm is \emph{not} parallelizable. Due to its greedy nature, when implemented as a parallel algorithm, this procedure requires $n$ adaptive rounds of queries.

Nevertheless, because of the broad applicability of matroid basis finding, Karp, Upfal, and Wigderson posed the following question in their seminal work \cite{KUW85, KUW88}:

\begin{center}
    \emph{What is the parallel complexity of finding a basis of an arbitrary matroid?}
\end{center}

Formally, \cite{KUW85, KUW88} ask for an algorithm which, given any matroid $\M = (E, \I)$, makes rounds of queries to the independence oracle, where each round is permitted $\mathrm{poly}(n)$ many queries in parallel, and finally outputs a basis of $\M$.

Because of the generality of matroids, this question is central in the study of parallel combinatorial optimization, with connections to seminal work on maximal independent sets~\cite{Lub86} and graph matchings~\cite{Lov79, KUW86, FGT16, ST17}, and more recent attention paid to submodular function minimization~\cite{BS20, CCK21, CGJS22}, submodular function maximization~\cite{BS18a,BS18b,BRS19a,BRS19b,CQ19a,CQ19b,EN19,ENV19,FMZ19a,FMZ19b,KMZ+19,CFK19,BBS20,LLV20}, and matroid intersection~\cite{GGR22, GT17, Bli22, BT25}.

\paragraph{Our Results}

Already in the work of \cite{KUW85, KUW88}, it was shown that there are algorithms which far outperform the greedy algorithm. Indeed, \cite{KUW85, KUW88} presented an algorithm which makes $\mathrm{poly}(n)$ queries to the independence oracle, yet computes a basis  of the matroid in only $O(\sqrt{n})$ many adaptive rounds. They paired this result with a lower bound, showing that there are classes of matroids (specifically, a class called partition matroids) which require $\Omega \left (\frac{n^{1/3}}{\log^{1/3}n} \right )$ many rounds of $\mathrm{poly}(n)$ queries to the independence oracle to find a basis. 

This gap between $\sqrt{n}$ and $\frac{n^{1/3}}{\log^{1/3}n}$ persisted for the ensuing four decades, until the recent work of Khanna, Putterman, and Song \cite{KPS25, KPS26}. In a sequence of works, \cite{KPS25} first showed parallel basis finding algorithms with complexity $\widetilde{O}(n^{7/15})$ rounds, and \cite{KPS26} presented a refinement which improved the round complexity to $\widetilde{O}(n^{3/7})$ rounds. Nevertheless, the polynomial gap between upper and lower bounds for this problem still remained.

As our primary contribution, we essentially close this gap:
\begin{theorem}\label{thm:1-3roundsintro}
There is an algorithm which for any arbitrary matroid $\mathcal{M}$ on $n$ elements, makes polynomially many independence queries per round and recovers a basis of $\mathcal{M}$ in $O(n^{1/3}\log^{1/3}n)$ rounds with high probability. 
\end{theorem}
Note that our result matches the lower bound of $\frac{n^{1/3}}{\log^{1/3}n}$ to within a $\log^{2/3}n$ factor, thereby showing that the right complexity for this problem is indeed $\widetilde\Theta(n^{1/3})$ and largely settling this line of research.\footnote{It is tempting to ask whether one can provide a tighter analysis of the lower bound instance of \cite{KUW85} to actually \emph{match} the upper bound we provide here. It turns out that this is not possible, as we show in \cref{sec:partitionMatroids} that there is an algorithm which finds bases of any partition matroid in $O(n^{1/3}/\log^{1/3}n)$ many rounds.}

On a technical level, our work introduces a new perspective for understanding how dependences emerge in a matroid under random sampling of the ground set. Our techniques are inspired by recent work on the mixing times of the matroid basis exchange walk (see, for instance, \cite{ALGV19, anari2021entropic, Lee24}), and provide a \emph{far tighter} understanding of the rates at which one can randomly sample without altering the structure of the matroid. We provide a more detailed overview of our techniques in the technical overview (see \cref{sec:techOverview}).

Our result has other immediate consequences too: one important quantity of interest in matroids is the \emph{rank} of the matroid, which is equivalent to the size of a basis. Thus, \cref{thm:1-3roundsintro} immediately implies that there is an $O(n^{1/3}\log^{1/3}n)$ round algorithm for computing the rank of any matroid. Via a simple, well-known reduction (see, for instance, \cite{BT25}), this then allows for efficient computation of the maximum or minimum weight basis in weighted matroids. More formally, a weighted matroid is given by $\M = (E, \I, w)$, where $w: E \rightarrow \R$. For a basis $B \subseteq E$, we say that the weight of $B$ is $w(B) = \sum_{e \in B} w(e)$. A basis has maximum (or minimum) weight if it has the largest (or smallest) weight among all bases of the matroid. As a formal corollary, we then have:

\begin{corollary}
    There is an algorithm which for any arbitrary, weighted matroid $\mathcal{M} = (E, \I, w)$ on $n$ elements, makes polynomially many independence queries per round and recovers a maximum (or minimum) weight basis of $\mathcal{M}$ in $O(n^{1/3}\log^{1/3}n)$ rounds with high probability. 
\end{corollary}

Our result can also be derandomized for many natural classes of matroids. One such prominent class of matroids is the class of \emph{linear matroids}. These are matroids $\mathcal{M} = (E, \I)$ such that there exists a field $\F_q$ and a set of vectors $E = \{v_1, \dots v_n\}$, such that a set $S \subseteq [n]$ is independent if and only if $v_i:i\in S$ is linearly independent. For this class of matroids, we show:

\begin{theorem}\label{thm:1-3roundsLinearIntro}
There is a (non-uniform) deterministic algorithm which for any linear matroid $\mathcal{M}$ on $n$ elements, makes polynomially many independence queries per round and recovers a basis of $\mathcal{M}$ in $O(n^{1/3}\log^{1/3}n)$ rounds.
\end{theorem}

Note that the lower bound of \cite{KUW85} does hold for the class of linear matroids even against randomized algorithms. The above upper bound shows that this lower bound can essentially be matched even with a deterministic algorithm.

\paragraph{Applications}

Naturally, because matroids have a host of applications in theoretical computer science, having faster parallel algorithms for finding matroid bases leads to better complexities in downstream applications. As remarked in \cite{KPS25, KPS26} one immediate application of more efficient basis finding is to the \emph{matroid intersection} problem. In this problem, one is given two matroids $\M_1 = (E, \I_1)$, $\M_2 = (E, \I_2)$, which are defined over the same ground set of elements. The goal is to find a set of elements $S \subseteq E$ which is independent in \emph{both} $\M_1$ and $\M_2$ and is of maximal size; i.e., $S \in \I_1 \cap \I_1$, and $|S|$ is as large as possible.

Matroid intersection has a long history of study. In the sequential (non-parallel) setting, \cite{Edm70} presented the first polynomial query algorithm, with recent improvements in the recent work of Chakrabarty et. al. \cite{CLS+19} which gave an $O(n \sqrt{r})$-rank-query algorithm for this task, and Blikstad \cite{Bli21} which gave an $O(n r^{3/4})$-independence-query algorithm. This problem has also seen significant work in the parallel setting. The work of Chakrabarty, Chen, and Khanna \cite{CCK21} established an $\Omega(n^{1/3})$-round lower bound for this problem, even in the more powerful setting of \emph{rank oracles}. \cite{Bli22} gave the first $o(n)$ round \emph{parallel} algorithms for this problem, with later improvements by \cite{BT25}, which yielded $O(n^{2/3})$ round rank-query algorithms for this problem, and $O(n^{5/6})$ round independence query algorithms.

By plugging in their improved basis-finding algorithms, \cite{KPS25, KPS26} further improved on this complexity by providing an $\widetilde{O}(n^{17/21})$ round independence query algorithm. In this same vein, our much improved basis finding algorithm yields a better algorithm for this problem:

\begin{theorem}\label{thm:matroidIntersection}
    There is an algorithm which for any arbitrary matroids $\mathcal{M}_1, \M_2$ on the same ground set of $n$ elements, makes polynomially many independence queries per round and solves the matroid intersection problem in $\tilde O(n^{7/9})$ rounds with high probability. 
\end{theorem}

Another setting where matroid basis finding is a key subroutine is for submodular optimization under a matroid constraint. In this problem, one is given a submodular function $f: 2^E \rightarrow \R_{\geq 0}$ along with a matroid $\M = (E, \I)$. The goal is to find the set $S \in \I$ which \emph{maximizes} $f(S)$. This problem naturally is NP-hard, and so works instead focus on \emph{approximately} maximizing this objective function. 

For sequential algorithms, the first polynomial time $(1 - 1/e)$-approximation came in the work of \cite{NWF78} for the restricted setting of \emph{cardinality constraints} on the set $S$. This was later generalized to arbitrary matroid constraints in the work of Vondr\'ak \cite{Von08}, and is known to be optimal among algorithms making $\mathrm{poly}(n)$ queries to the function $f$. 

In the parallel (sometimes called adaptive) setting, the study of this problem was initiated by Balkanski and Singer \cite{BS18a}. A sequence of works \cite{BRS19a, CQ19a, CQ19a, EN19, FMZ19a, FMZ19b, ENV19} established that there is a $\mathrm{poly}(\log(n), 1 / \eps)$-round $(1 - 1/e - \eps)$ approximation algorithm for this task, when the matroid is accessed via \emph{rank} queries. When the matroid is given via independence queries, the best complexity was originally due to the work of \cite{BRS19b} which presented an $O(\sqrt{n}/\eps^3)$ round algorithm. This was later improved in the work of \cite{KPS26} which provided an $\widetilde{O}(n^{3/7}/\eps^3)$ round algorithm.

As observed in \cite{KPS26}, any improvements in parallel basis finding immediately yield improvements for submodular optimization. By plugging in our improved basis finding algorithm into the framework of \cite{BRS19b}, we immediately obtain the following theorem:

\begin{theorem}\label{thm:MonotoneSubmodularOpt}
    For any $\eps > 0$, there is an $\widetilde{O}(n^{1/3}/ \eps^3)$ round parallel algorithm which, for any matroid $\M = (E, \I)$ and any monotone submodular function $f: 2^E \rightarrow \R_{\geq 0}$, makes polynomially many independence queries to $\M$ and polynomially many queries to the function $f$, and with high probability, outputs a $(1 - 1/e - \eps)$ approximation to the maximum of $f$ under $\M$.
\end{theorem}

\paragraph{Outline}

In \cref{sec:techOverview}, we give a high-level overview of the techniques underlying prior work \cite{KUW85, KUW88, KPS25, KPS26}, as well as the key insights and techniques we introduce to design a near-optimal parallel algorithm. In \cref{sec:preliminaries}, we introduce basic facts about matroids and probability that we use throughout the paper. In \cref{sec:densityPreservation} we prove a key fact about subsampling dense set families. In \cref{sec:decomposition}, we formally recap the decomposition statement of \cite{KPS25}, and extract the necessary pieces that we will use. In \cref{sec:deletion}, we show a much better mechanism for identifying redundant elements, which we then use to design our basis finding algorithm in \cref{sec:optalgorithm}.

\section{Technical Overview}\label{sec:techOverview}

\subsection{Prior Work of Karp, Upfal, and Wigderson}

We start by recapping the work of Karp, Upfal, and Wigderson \cite{KUW85, KUW88}. At a high level, their work introduces two new methods for making ``progress'' towards finding a basis of a matroid $\M = (E, \I)$:

\begin{enumerate}
    \item \textbf{Contraction:} The first method of making progress towards finding a matroid basis is to find a set $S \subseteq E$ of elements which is \emph{independent}; i.e., such that $S \in \I$. By the extension property of matroids, we know that there must exist a basis of $\M$ which contains $S$. Thus, we commit to keeping the set $S$ in our future basis, and then instead focus on finding a basis of the \emph{contracted} matroid $\M/S$. This contains elements $E - S$ and a set $S' \subseteq E - S$ is independent if and only if $S' \cup S \in \I$. 
    \item \textbf{Deleting Redundant Elements:} The second method of making progress towards finding a matroid basis is to find a set of elements $S \subseteq E$ such that, for each element $e \in S$, $e \in \mathrm{span}(E - S)$. In this way, each element $e$ comes with a ``certificate'' of its redundancy. One can then delete this set of elements $S$, and instead find a basis of $\M|_{E-S} = (E-S, \I \cap 2^{E -S})$. Note that because for each element $e \in S$, $e \in \mathrm{span}(E - S)$, it is still the case then that $\mathrm{rank}(\M|_{E-S}) = \mathrm{rank}(\M)$, and so any basis of $\M|_{E-S}$ is also a basis of $\M$.
\end{enumerate}

With these two methods of making progress established, \cite{KUW85, KUW88} then propose a simple two-case procedure for finding a matroid basis. Indeed, given their matroid $\M = (E, \I)$, they partition the ground set $E$ into $\sqrt{n}$ groups of $\sqrt{n}$ elements, which we call $S_1, \dots , S_{\sqrt{n}}$. Within each bucket $S_i$, their algorithm queries every prefix to the independence oracle (so it makes the queries $\{e_1\}, \{e_1, e_2\}, \dots \{e_1, e_2, \dots, e_{\sqrt{n}}\}$). There are then only two cases to understand:
\begin{enumerate}
    \item Suppose that for some group $S_i$, every prefix query is independent. This then means that the entire group is independent, i.e., that $\mathrm{Ind}(S_i) = 1$. By the first item above, \cite{KUW85} observe that this means there is an independent set of size $\sqrt{n}$ that we can contract on, and then recursively solve $\M/S_i$.
    \item Otherwise, for \emph{every group} $S_i$, it must be the case that some prefix queries are dependent. \cite{KUW85} then look at the first instance in which the prefix queries go from independent to dependent; i.e., say that $\Ind(\{e_1, \dots e_{j-1}\}) = 1$, but $\Ind(\{e_1, \dots e_{j-1}, e_j\}) = 0$. This then implies that $e_j \in \mathrm{span}(\{e_1, \dots e_{j-1}\})$. Because this can be done for each group $S_i$ \emph{in parallel}, this leads to a total of $\sqrt{n}$ elements that can be deleted.
\end{enumerate}

Thus, in just one single round of $\mathrm{poly}(n)$ queries (simply the prefix queries) \cite{KUW85} either contract on an independent set of $\sqrt{n}$ elements or delete a set of $\sqrt{n}$ redundant elements. A simple calculation shows that in just $O(\sqrt{n})$ rounds, this suffices for computing a basis of any matroid. 

\subsection{The Work of Khanna, Putterman, and Song}

More recently, the work of Khanna, Putterman, and Song \cite{KPS25, KPS26} provided the first improvement over the $O(\sqrt{n})$ round algorithm of \cite{KUW85}. At a high level, their algorithm operates less greedily than that of \cite{KUW85}; rather than contracting on independent sets or deleting redundant elements after \emph{every round}, \cite{KPS25, KPS26}  instead invest many rounds of queries to decompose the matroid $\M$ into ``well-behaved'' pieces. Only once this decomposition is completed do \cite{KPS25, KPS26} begin to find independent sets and redundant elements.

\paragraph{The Decomposition of \cite{KPS25, KPS26}}

In order to properly describe the decomposition of \cite{KPS25, KPS26}, we must first introduce some terminology. A central component of their decomposition is the notion of a circuit induced by a permutation; given a matroid $\M = (E, \I)$, they consider sampling a random permutation $\pi:[n]\to E$. They then query each prefix of $\pi$ to the independence oracle, namely $\Ind(\{\pi(1)\}), \Ind(\{\pi(1), \pi(2)\}), \dots \Ind(\{\pi(1), \dots ,\pi(n)\})$. Assuming the matroid $\M$ is not full-rank, eventually these queries \emph{must} become dependent, and moreover, the first index where these queries become dependent will produce a \emph{unique} circuit, which they denote $C_{\pi}$. Note that here, a circuit refers to a minimal set of elements which is dependent. The above implies that if $\Ind(\{\pi(1), \dots ,\pi(i)\})$ is the first prefix query which is dependent, then there is in fact a \emph{unique} smallest set of elements $C_{\pi}\subseteq \{\pi(1), \dots ,\pi(i)\}$ which is dependent.

\cite{KPS25, KPS26} now repeat this experiment of sampling a random permutation and finding $C_{\pi}$ many times in parallel. In so doing, they are able to estimate several parameters of the the matroid $\M$:
\begin{enumerate}
    \item $\alpha(\M)$: this is the \emph{median index} at which the circuit $C_{\pi}$ forms.
    \item $p_{x, \M}$: for an element $x \in E$, this is the marginal probability that $x \in C_{\pi}$. I.e., $p_{x, \M} = \Pr_{\pi \sim \M}[x \in C_{\pi}]$.
    \item $q_{T}$: for a set $T \subseteq E$, this is the circuit containment probability, or the probability that $C_{\pi}$ is contained in $T$. Formally, $q_{T} = \Pr_{\pi}[C_{\pi} \subseteq T]$.
\end{enumerate}

In particular, \cite{KPS25, KPS26} show that in a single round of only $\mathrm{poly}(n)$ many independence queries, they can find a set $S$ such that $q_{S}\geq 0.99$, and yet simultaneously for every element $x \in S$, $p_{x, \M|_S} \geq \frac{1}{n}$. They then peel off this set $S$, and then \emph{repeat} this procedure on the matroid $\M|_{E - S}$. \cite{KPS25, KPS26} show that these peeled sets satisfy numerous convenient properties. For instance, one property of these peeled sets which is important is that $\frac{\alpha(S)}{|S|} \approx \frac{\alpha(\M)}{|\M|}$. 

Perhaps even more important is that these $\alpha$-values of the peeled sets strictly increase at meaningful rates. Indeed, the key lemma that \cite{KPS25, KPS26} show is that if two sets $S_i, S_j$ are peeled off at rounds $i < j$ of the decomposition and $|S_i| \approx |S_j|$, it must be the case that $\alpha(S_j) \geq \alpha(S_i) + \Omega(\sqrt{\alpha(S_i)})$. Intuitively, this means that if we focus on sets of a fixed size $\beta$, then after peeling $\sqrt{\beta}$ many sets of this size, it must be that the alpha-value has increased to $\Omega(\beta)$. This turns out to suffice for bounding the number of rounds that the decomposition can last for; after $n^{1/3}$ many rounds, the $\alpha$-values will have grown to $n^{2/3}$. Because a set is always larger than its $\alpha$-value, every set peeled from this point on is of size $\geq n^{2/3}$, and so there can only be $n^{1/3}$ many of them. All together, this gives the following decomposition statement:

\begin{lemma}[Informal, for intuition only.]\label{lem:informalDecomposition}
    Let $\M = (E, \I)$ be a matroid. There is a decomposition algorithm running in $k = O(n^{1/3})$ many rounds which produces a partition $S_1, \dots S_k$ of $E$ such that:
    \begin{enumerate}
        \item For each $i \in [k]$ and $x \in S_i$, $p_{x, \M|_{S_i}} \geq \frac{1}{n}$.
        \item When set $S_i$ is peeled off, it is the case that $\frac{\alpha(S_i)}{|S_i|} \approx \frac{\alpha(\M - S_1 - \dots -S_{i-1})}{|\M - S_1 - \dots -S_{i-1}|}$.
    \end{enumerate}
\end{lemma}

With this decomposition in hand, the key question then becomes how to use it to efficiently find a basis. 

\paragraph{Subroutines for Contraction}

The starting observation of both of \cite{KPS25, KPS26} is simple: namely, if the $\alpha$-value of a set $S_i$ ever becomes large then there is a simple way to make progress by finding a large independent set. Indeed, the second point of \cref{lem:informalDecomposition} immediately implies that $\frac{\alpha(\M - S_1 - \dots -S_{i-1})}{|\M - S_1 - \dots -S_{i-1}|} \approx \frac{\alpha(S_i)}{|S_i|}$, and so $\alpha(\M - S_1 - \dots -S_{i-1}) \geq |\M - S_1 - \dots -S_{i-1}| \cdot \frac{\alpha(S_i)}{|S_i|}$. Note that the $\alpha$-value of $\M - S_1 - \dots -S_{i-1}$ measures the \emph{median} index at which the first dependence (circuit) forms when sampling a random permutation $\pi$. So in particular, by sampling many random permutations, one can find an independent set of size $\geq \alpha(\M - S_1 - \dots S_{i-1}) $. For simplicity, we will assume that $|\M - S_1 - \dots -S_{i-1}| = \Omega(n)$, in this case the above implies that we can find an independent set of size $\Omega \left (n \cdot \frac{\alpha(S_i)}{|S_i|} \right)$.

\paragraph{Subroutines for Deletion}
The above discussion provides a simple way to make progress towards finding a matroid basis whenever $\frac{\alpha(S_i)}{|S_i|}$ becomes close to $1$. Naturally however, there is no guarantee that this will be the case when running the decomposition algorithm. Thus, the remaining ingredients from \cite{KPS25, KPS26} are sub-routines for finding \emph{redundant elements} whenever $\alpha(S_i) \ll |S_i|$.

In \cite{KPS25}, the main deletion sub-routine carefully samples random sets to find redundant elements, and ultimately shows that in a single additional round of queries, one can find $\widetilde{\Omega }\left ( \min \left (|S_i|, \frac{|S_i|^2}{\alpha(S_i)^2} \right )\right )$ many redundant elements. When $\alpha(S_i) \leq \sqrt{|S_i|}$, this shows that one additional round of queries suffices for deleting an $\widetilde{\Omega}(1)$ fraction of the elements from a set $S_i$. As $\alpha(S_i)$ grows larger than $\sqrt{|S_i|}$, this no longer holds true, but nevertheless, this sub-routine provides a type of win-win paradigm based on the $\alpha$-value for a set $S_i$: whenever $\alpha(S_i)$ is large, one can find a large independent set in the parent matroid, and whenever $\alpha(S_i)$ is small, one can find more redundant elements to delete from $S_i$.

This tradeoff is further refined in the work of \cite{KPS26}. Here, the authors exploit additional properties from a refinement of \cref{lem:informalDecomposition} to show that in a single additional round, they can find $\widetilde{\Omega }\left ( \min \left (|S_i|, \frac{|S_i|^{3/2}}{\alpha(S_i)} \right )\right )$ many redundant elements to delete. This, in combination with a more ``dynamic'' analysis of when to delete redundant elements vs. when to contract on independent elements then leads to the state of the art complexity of $\widetilde{O}(n^{3/7})$ many rounds for finding a matroid basis. 

It is worth thinking about why these algorithms of \cite{KPS25} and \cite{KPS26} fail to achieve a round complexity of $\widetilde{O}(n^{1/3})$. Indeed, even though the decomposition algorithm runs in $\widetilde{O}(n^{1/3})$ rounds, there are choices of the $\alpha$-value for which the progress made (i.e., the total number of deletions or the total size of the contracted set) is bounded away from $n$. For instance, imagine peeling off a set $S_i$ of size $\Omega(n)$, where $\alpha(S_i) = |S_i|^{3/4}$. \cref{lem:informalDecomposition} only guarantees that you can find an independent set of size $\approx n^{3/4}$, while even the best deletion sub-routine from \cite{KPS26} only guarantees that one can find $\widetilde\Omega\left ( \frac{n^{3/2}}{n^{3/4}}\right ) = \widetilde\Omega\left ( n^{3/4}\right )$ many redundant elements. Thus, after investing $O(n^{1/3})$ many rounds in the decomposition, the total progress made towards a matroid basis is still some $n^{1 - \Omega(1)}$, and so the decomposition must be repeated $n^{\Omega(1)}$ times.\footnote{While this intuition is correct, the actual tradeoff in \cite{KPS25, KPS26} is slightly different. In these works, the decomposition running for $n^{1/3}$ rounds is actually a \emph{good case}. The bad case is when the decomposition runs for some $n^{1/3 - \eps}$ many rounds for $\eps$ a small constant, and then reveals a set $S_i$ with the aforementioned characteristics.} This ultimately leads to complexities of $n^{1/3 + \Omega(1)}$ and has constituted the fundamental barrier towards matching the $\tilde \Omega(n^{1/3})$ round lower bounds from \cite{KUW85}.

\subsection{A New, Near-Optimal Deletion Subroutine}

Our primary contribution is a $1$-round deletion procedure which \emph{far} outperforms those of \cite{KPS25, KPS26}:

\begin{theorem}[Informal]\label{lem:deletionIntro}
    Let $S_i \subseteq E$ be such that $p_{x, \M|_{S_i}} \geq \frac{1}{n}$ and $\alpha(S_i) \leq \frac{|S_i|}{100 \log(n)}$. Then, there is a one round algorithm (making polynomially many independence queries) which finds a set of $\Omega(|S_i| / \log(n))$ many redundant elements. 
\end{theorem}

In the regime where $\alpha(S_i)= \frac{|S_i|}{100 \log(n)}$, the best previous deletion subroutine is that of \cite{KPS26}, which finds $\widetilde{\Omega}(\sqrt{|S_i|})$ many redundant elements. Our new algorithm completely closes the gap, efficiently finding $\Omega(|S_i| / \log(n))$ redundant elements, compared to the maximum possible of $|S_i|$. Combined with the aforementioned decomposition, this immediately yields a basis-finding algorithm which runs in $\widetilde{O}(n^{1/3})$ many rounds with a far simpler analysis than that of \cite{KPS25, KPS26}.
We defer the discussion of this complete picture to the next section (see \cref{sec:techOverviewFinalAlg} for this).

\paragraph{The Deletion Algorithm}

Now, we proceed to give intuition about the deletion algorithm. As a reminder, we know that our set $S_i \subseteq E$ satisfies that for all $x \in S_i$, $p_{x, \M|_{S_i}} \geq \frac{1}{n}$ and that $\alpha(S_i) \leq \frac{|S_i|}{100 \log(n)}$. 
In particular, this means that if we sample many random permutations $\pi$ over the set $S_i$, then it will be the case that every element will appear in some of the resulting circuits $C_{\pi}$. It is tempting to say that this is already enough to find many redundant elements; indeed, if, for every element, we can find a circuit which uses that element, then why does this not suffice for our ``certificate'' of redundancy?

Unfortunately, the key property that we need when deleting a redundant set of elements $R \subseteq S_i$ is that $R \subseteq \mathrm{span}(S_i - R)$. In this way, we can ensure that \emph{even after} $R$ is deleted, the rank of the matroid is unchanged. So, even though we may recover many circuits which collectively contain every element, if we ever delete an element $x \in S_i$ (equivalently, add it to $R$), it is effectively used up, and any circuit $C_{\pi}$ which contains $x$ is no longer a valid certificate! This becomes the primary difficulty in designing an efficient deletion subroutine: we must find a (large!) set of elements $R$ such that \emph{all} of the certificates of redundancy for the elements in $R$ are completely contained in $S_i - R$.

Our algorithm for doing this is quite simple, although it should not be obvious a priori that our algorithm will work. Given the set $S_i$, we will sample a set $S' \subseteq S_i$ which contains exactly $m = |S_i| \cdot (1 - 1 / \log(n))$ many elements. Intuitively, this set $S'$ is obtained by throwing away a $1 / \log(n)$ fraction of the elements. $S'$ itself still contains the majority of elements. Now, for the remaining elements $x \in S_i - S'$, we try to find a circuit which contains the element $x$ but otherwise uses \emph{only} elements from $S'$. We do this by repeating our random permutation experiment; we sample permutations $\pi$ over $S' \cup \{x\}$, and add elements in this order until a circuit forms. The key theorem that we show is that for an $\Omega(1)$ fraction of the elements $x \in S_i - S'$, this experiment will reveal circuits that contain the element $x$! In this way, we are able to find $\Omega(|S_i - S'|) = \Omega(|S_i| / \log(n))$ many redundant elements in $S_i - S'$, as $S'$ contains a ``certificate'' of redundancy for each of them.

\paragraph{Analyzing the Deletion Algorithm}
Analyzing the above algorithm is very subtle. We know that in the set $S_i$, all elements have marginal probabilities $p_{x, \M|_{S_i}} \geq \frac{1}{n}$. But the algorithm is essentially relying on how marginal probabilities \emph{change} under mild subsampling of the set $S_i$. 

Let us fix an element $x \in S_i$. Because we are doing mild subsampling, it may be tempting to say that any circuit $C_{\pi}$ which appears under random permutations $\pi$ and contains $x$ will still be likely to appear after subsampling at rate $1 - 1 / \log(n)$. This is very strongly \emph{not} the case. A single circuit $C_{\pi}$ may be polynomially long; as an example, its size may be $|S_i|^{1/2}$. Under subsampling, such a circuit will only remain intact if \emph{all} of its elements are contained in the sample at rate $(1 - 1 / \log(n))$. But in this example, this only occurs with probability $(1 - 1 / \log(n))^{|S_i|^{1/2}} = (1/2)^{n^{\Omega(1)}}$ when $|S_i| = n^{\Omega(1)}$, which is extremely small. We thus require a more robust understanding of how circuits arise in $S_i$.

To do this, we re-interpret the $\alpha$-value condition: because we assume that $\alpha(S_i) \leq \frac{|S_i|}{100 \log(n)}$, this means that when we do the random permutation experiment on $S_i$, it is \emph{very unlikely} for the first circuit to arise only after position $|S_i|/2$. One way to see this is that a random sample of $\frac{|S_i|}{100 \log(n)}$ creates a circuit with probability $\geq 1/2$. Repeating this $50 \log(n)$ times then means that $\leq |S_i|/2$ elements are sampled, and the probability that there is no circuit in the sampled elements is $\leq (1/2)^{50 \log(n)} \leq \frac{1}{n^{50}}$. Next, let us fix an element $x \in S_i$. By our marginal condition, we know that $\Pr[x \in C_{\pi}] \geq \frac{1}{n}$. Combined with the above, we know that $\Pr_{\pi}[x \in C_{\pi} \text{ and } C_{\pi} \text{ forms before position }|S_i|/2] \geq \frac{1}{n} - \frac{1}{n^{50}} \geq \frac{1}{2n}$. We can now refine this even further: for this element $x$, we consider an index $\ell \in [|S_i|/2]$, and by a simple invocation of the pigeonhole principle, we know there is a choice of $\ell$ such that $\Pr_{\pi}[x \in C_{\pi} \text{ and } C_{\pi} \text{ forms at position }\ell] \geq \frac{1}{n^2}$.

In the above discussion, we saw that the circuits $C_{\pi}$ themselves change significantly over the course of (even mild) random sampling. Thus, we do one final re-write of the above expression. Rather than looking at the circuit $C_{\pi}$, we consider the set $T$ consisting of the first $\ell$ elements in the permutation $\pi$, excluding the element $x$. In the event that $C_\pi$ forms at position $\ell$ and $x\in C_\pi$, then we \emph{know} that $T = C_{\pi} - \{x\}$ must be independent and that $x\in \mathrm{span}(T)$. An alternative formulation of the above statement is then that, for our element $x$,
\[
\Pr_{T \sim \binom{S_i - \{x\}}{\ell-1}}[x \in \mathrm{span}(T) \text{ and } \Ind(T) = 1] \geq \frac{1}{n^2}.
\]

In words, let us say that a set $T$ of size $\ell-1$ is \emph{good} for element $x$ if $x \in \mathrm{span}(T) \text{ and } \Ind(T) = 1$. Intuitively, such a set $T$ is a \emph{valid witness of redundancy} for the element $x$, as if one ever queries $\Ind(T)$ and $\Ind(T \cup \{x\})$ one receives a simple certificate that $x \in \mathrm{Span}(T)$. The above statement is saying that in our original set $S_i$, a $\geq \frac{1}{n^2}$ fraction of the sets $T \in \binom{S_i-\{x\}}{\ell-1}$ are \emph{good} for the element $x$. This then gives us a new perspective for understanding our proposed subsampling deletion algorithm: can we show that, for every element $x$, after subsampling a $(1 - 1 / \log(n))$ fraction of elements, that the fraction of sets which are good for this element $x$ is still large?

\paragraph{Density Preservation Under Subsampling}

Our key result answers the previous question in the affirmative. I.e., we show that mild subsampling at rate $(1 - 1/\log(n))$ does (essentially) preserve the fraction of good sets for any element $x$! We show the following:

\begin{theorem}[Informal]\label{thm:densitySubsampleTechOverview}
    Let $S$ be a set of $\leq n$ elements and let $\ell \leq |S|/2$. Let $m \geq  |S| \cdot (1 - 1/\log(n))$. Let $\mathcal{F} \subseteq \binom{S}{\ell}$ be a collection of subsets of $S$ of size $\ell$, and define 
    \[
    p = \frac{|\mathcal{F}|}{\binom{S}{\ell}} \geq \frac{1}{n^2}.
    \]
    Now, sample a set $S' \subseteq S$ with exactly $m$ elements. Then, 
    \[
    \Pr_{S'} \left [ \frac{\left|\mathcal{F} \cap \binom{S'}{\ell} \right|}{\binom{m}{\ell}} \geq \frac{p}{2}\right ] = \Omega(1).
    \]
\end{theorem}

To invoke \cref{thm:densitySubsampleTechOverview}, we let $\mathcal{F} \subseteq \binom{S_i-\{x\}}{\ell-1}$ denote the set of \emph{all good sets} for an element $x$ of our choosing. The preceding discussion shows that the density condition is indeed satisfied. The conclusion of the theorem is then that, after subsampling, there is an $\Omega(1)$ probability (over the sampled set $S'$) that the set of good sets for $x$ is \emph{still dense} in the new set $S'$. Thus, when we run the permutation experiment in $S' \cup \{x\}$, it will still be the case that we find many circuits containing $x$.

Surprisingly, proving \cref{thm:densitySubsampleTechOverview} relies on techniques that originate from the recent work of \cite{anari2021entropic}, which in turn is related to the breakthrough work of \cite{ALGV19} on the mixing times of the matroid basis exchange walk. To describe these connections, we require the notion of an up-operator: starting with a distribution $\mu$ over $\binom{S}{\ell}$, the up-operator $U_{\ell \rightarrow \ell+1}$ maps $\mu$ to a distribution over $\binom{S}{\ell+1}$ by ``spreading out'' the mass of each set. I.e., for each set $T \in \binom{S}{\ell}$, $U_{\ell \rightarrow \ell+1}(\mu)$ takes the probability mass from $\mu(T)$ and divides it uniformly among all sets $T'$ of size $\ell+1$ such that $T \subseteq T'$. In this way $U_{\ell \rightarrow \ell+1}(\mu)(T') = \sum_{T \subseteq T', |T| = \ell} \frac{\mu(T)}{|S| - \ell}$.

The key fact from \cite{anari2021entropic} (with the exact version we use being attributed to \cite{Lee24}) is that the up-operator decreases the KL-divergence between a distribution and the uniform distribution. I.e., $d_{\mathrm{KL}}(U_{\ell \rightarrow \ell+1}(\mu), \mathcal{U}_{\ell+1}) \leq \frac{|S| - \ell -1}{|S| - \ell} \cdot d_{\mathrm{KL}}(\mu, \mathcal{U}_{\ell})$. After repeated application, one can then show that $d_{\mathrm{KL}}(U_{\ell \rightarrow m}(\mu), \mathcal{U}_{m}) \leq \frac{|S| - m}{|S| - \ell} \cdot d_{\mathrm{KL}}(\mu, \mathcal{U}_{\ell})$. 

In our setting, we let $\mu$ be the uniform distribution over $\mathcal{F}$, and we let $m =|S| \cdot (1 - 1/\log(n)) $. This then implies that $d_{\mathrm{KL}}(U_{\ell \rightarrow m}(\mu), \mathcal{U}_{m}) \leq \frac{2}{\log(n)} \cdot d_{\mathrm{KL}}(\mu, \mathcal{U}_{\ell})$. By direct calculation, one can also show that $d_{\mathrm{KL}}(\mu, \mathcal{U}_{\ell}) = \log(1/p) \leq 2\log(n)$. Thus, we obtain that \begin{align}\label{eq:boundKLIntro}
d_{\mathrm{KL}}(U_{\ell \rightarrow m}(\mu), \mathcal{U}_{m}) \leq 4.
\end{align}

It is worth reflecting on what the distribution $U_{\ell \rightarrow m}(\mu)$ actually is: this distribution places proportionally \emph{more probability mass} on sets $S' \in \binom{S}{m}$ which have a \emph{larger} intersection with $\mathcal{F}$. Likewise those sets $S'$ which have a smaller intersection with $\mathcal{F}$ receive less mass in $U_{\ell \rightarrow m}(\mu)$. A formal accounting of this phenomenon reveals an important fact: if one defines $\mathcal{B}$ to be the set of all $S'$'s which have a ``small'' intersection with $\mathcal{F}$, i.e.,  $\mathcal{B} = \left \{S' \in \binom{S}{m}: \frac{\left|\mathcal{F} \cap \binom{S'}{\ell} \right|}{\binom{m}{\ell}} < \frac{p}{2} \right\}$, then one can show by the above logic that  $U_{\ell \rightarrow m}(\mu)(\mathcal{B}) < 1/2.$ Importantly then, for $\mathcal{G} = \overline{\mathcal{B}}$, we have that
\begin{align}\label{eq:boundBadMassIntro}
U_{\ell \rightarrow m}(\mu)(\mathcal{G}) \geq 1/2.
\end{align}
Together, \cref{eq:boundKLIntro} and \cref{eq:boundBadMassIntro} suffice for proving our desired theorem: indeed because $d_{\mathrm{KL}}(U_{\ell \rightarrow m}(\mu), \mathcal{U}_{m})$ is small, and $U_{\ell \rightarrow m}(\mu)$ places a lot of probability mass on the event $\mathcal{G}$, this then means that when one samples $S'$ in accordance with $\mathcal{U}_{m}$, it must \emph{still} be the case that $\mathcal{G}$ occurs with constant probability, which is exactly what we sought to prove.

\subsection{Putting the Algorithm Together}\label{sec:techOverviewFinalAlg}

It only remains to put the pieces together. We start by running the decomposition of \cref{lem:informalDecomposition}. If, at any point, we reach a stage where $\frac{\alpha(S_k)}{|S_k|} \geq \frac{1}{100 \log(n)}$, then we immediately terminate the decomposition and find an independent set in $\M - S_1 - S_2 - \dots - S_{k-1}$. Suppose $|\M - S_1 - S_2 - \dots - S_{k-1}|=\Omega(n)$, then we can recover and contract an independent set of size $\Omega (n/ \log(n) )$.

Otherwise, for every set $S_1, \dots ,S_{k}$ that we peel off that does satisfy the $\alpha$-value bound, we invoke the deletion sub-routine outlined above. For each $i\in [k]$, because $\frac{\alpha(S_i)}{|S_i|} < \frac{1}{100 \log(n)}$, \cref{lem:deletionIntro} guarantees that we can find $\Omega \left ( \frac{|S_i|}{\log(n)}\right )$ many redundant elements in just a single additional round of queries. 

Thus, we invest $O(n^{1/3})$ many rounds into our decomposition, and the result is that we either:
\begin{enumerate}
    \item Recover an independent set of size $\Omega \left (n / \log(n) \right )$ that we can contract on.
    \item Recover $\Omega \left ( \sum_{i= 1}^{k} |S_i|/\log(n) \right )=\Omega(n/\log (n))$ many redundant elements that we can delete.  
\end{enumerate}
In either case, we thus make $\Omega(n / \log(n))$ progess towards finding a matroid basis.

The round complexity of the algorithm is then $T(n) = T(n - \Omega( n / \log(n)) + O(n^{1/3})$, which can be solved to show that $T(n) = O(n^{1/3}\log(n))$. In conjunction with the $\Omega(n^{1/3} / \log^{1/3}(n))$ round lower bound of \cite{KUW85}, this effectively settles the round complexity of finding a matroid basis. Note that we in fact show a \emph{stronger} result in the body of the paper, namely an $O(n^{1/3}\log^{1/3}(n))$ round algorithm for finding a matroid basis. These further refinements are due to a more careful analysis of the matroid decomposition algorithm, which we omit from here for the sake of brevity.

\paragraph{Organization}

As mentioned before, in the remainder of the paper, we first introduce some basic notions and definitions for matroids (\cref{sec:preliminaries}). We then present a formal proof of our density preservation characterization (\cref{sec:densityPreservation}) and our refined decomposition algorithm (\cref{sec:decomposition}). Finally, in \cref{sec:deletion}, we show how the density preservation ideas can be used to design a near-optimal deletion algorithm and finally in \cref{sec:optalgorithm}, we present our optimal basis finding algorithm.

\section{Preliminaries}\label{sec:preliminaries}

We now present definitions and notions from prior work that will be used throughout our paper.

\subsection{Matroid Notions}

\begin{definition}[Matroids]
A \emph{matroid} $\M=(E,\I)$ is a pair where $E$ is a finite ground set and $\I \subseteq 2^{E}$ is a collection of independent sets with the following properties: (i) $\emptyset \in \I$ (non-triviality), (ii) for every $S\in \I$ and $S'\subset S$, $S'\in \I$ (downward-closedness), and (iii) for every $S,S'\in \I$ and $|S'|<|S|$, there exists some $x\in S\setminus S'$ such that $S+x\in \I$ (exchange property).
\end{definition}

\begin{definition}[Independent Sets, Circuits, Bases]
For a matroid $\M=(E,\I)$, we say a set $S\subseteq E$ is \emph{independent} if $S\in \I$ and \emph{dependent} otherwise. We call a set $B$ a \emph{basis} if it is a maximal independent set, i.e. for any $x\notin B$, $B+x\notin \I$. We call a set $C$ a \emph{circuit} if it is a minimal dependent set, i.e. for any $x\in C$, $C-x\in \I$.
\end{definition}

\begin{definition}[Rank]
For a matroid $\M=(E,\I)$, we define the \emph{rank} of $\M$ as $\rank(\M) = \max_{S\in \I} |S|$. Further, for any $S\subseteq E$, we define $\rank_{\M}(S) = \max_{T\subseteq S,T\in \I} |T|$. The rank function of a matroid is submodular.
\end{definition}

\begin{definition}[Span]
In a matroid $\M=(E,\I)$, we define $\text{span}(S)$ as
\[
\text{span}(S) = \{x\in E \mid \text{rank}(S\cup \{x\})=\text{rank}(S)\}.
\]
\end{definition}

\begin{definition}[Restriction, Contraction]
Let $\M=(E,\I)$ be a matroid and $S\subseteq E$. We write $\M|_S$ for the restriction of $\M$ to $S$, and we use $M - S$ to mean the restriction of $M$ to $E \setminus S$, and $\M/S$ for the contraction of $\M$ by $S$, whose rank function is $\rank_{\M/S}(T)=\rank_{\M}(S\cup T)-\rank_\M (T)$ for any $T\subseteq E\setminus S$.
\end{definition}

The following definitions come from the work of \cite{KPS25}.

\begin{definition}[Permutation-Induced Circuit]
    Let $\M=(E,\I)$ be a matroid, and let $\pi:[n]\to E$ be an arbitrary permutation. We let $C_{\pi}$ denote the unique first circuit which appears when adding elements in the order of $\pi$. I.e., if $\{\pi(1), \dots ,\pi(j)\} \in \I$, but $\{\pi(1), \dots \pi(j), \pi(j+1)\} \notin \I$, we let $C_{\pi}$ denote the unique circuit in $\{\pi(1), \dots \pi(j), \pi(j+1)\}$.
\end{definition}

\begin{definition}[$\alpha$-value]
For a matroid $\M=(E,\I)$ and $S\subseteq E$, we let $
    \alpha(S)$ denote the median number of elements sampled from the $S$ before a circuit forms. I.e., 
    \[
    \alpha(S) = \min \left \{k \in \N: \Pr_{T\sim {S\choose k}} [\mathrm{Ind}(T)=1]\leq \frac{1}{2} \right \}.
    \]
\end{definition}

\begin{definition}[Marginal Probability]
    Let $\M=(E,\I)$ be a matroid over $n$ elements. For an element $x \in E$, we let
    \[
    p_{x, \M} = \Pr_{\pi}[x \in C_{\pi}],
    \]
    where $C_{\pi}$ is the unique first circuit that appears when adding elements in the order of $\pi$. Similarly, for a subset $S\subseteq E$, we let
\[
q_S= \Pr_\pi[C_\pi\subseteq S].
\]
\end{definition}

\subsection{KL Divergence and Up and Down Operators}

In order to prove our density preservation statement under subsampling, we require the basic notions of KL-Divergence and up and down operators (see, for instance, \cite{Lee24} for further discussion).

\begin{definition}
    For two distributions $\mu, \nu$ over a set $\mathcal{X}$, their KL-divergence is defined as
    \[
    d_{\mathrm{KL}}(\mu \Vert \nu) = \sum_{x \in \mathcal{X}} \mu(x) \cdot \log \left ( \frac{\mu(x)}{\nu(x)}\right ).
    \]
\end{definition}

We will also use the notion of an \emph{up-operator} when dealing with distributions over sets:

\begin{definition}
    Let $k < n$ be integers, and let $\mu$ be a distribution over $\binom{[n]}{k}$. We define $U_{k \rightarrow k+1}(\mu)$ to be the distribution over $\binom{[n]}{k+1}$ such that for $S \in \binom{[n]}{k+1}$,
    \[
    U_{k \rightarrow k+1}(\mu)(S) = \sum_{T \subseteq S: |T| = k} \mu(T) \cdot \frac{1}{n-k}.
    \]
    Intuitively, for each set $T \in \binom{[n]}{k}$, the up-operator replaces $T$ with the uniform distribution of sets of size $k+1$ that contain $T$. 
\end{definition}

Likewise, we can define the down operator:

\begin{definition}
    Let $0<k \leq n$ be integers, and let $\mu$ be a distribution over $\binom{[n]}{k}$. We define $D_{k \rightarrow k-1}(\mu)$ to be the distribution over $\binom{[n]}{k-1}$ such that for $S \in \binom{[n]}{k-1}$,
    \[
    D_{k \rightarrow k-1}(\mu)(S) = \sum_{T: S \subseteq T, |T| = k} \mu(T) \cdot \frac{1}{k}.
    \]
    Intuitively, for each set $T \in \binom{[n]}{k}$, the down-operator replaces $T$ with the uniform distribution of sets of size $k-1$ contained in $T$.
\end{definition}

\begin{remark}
    Note that in general, we use the notation $U_{k \rightarrow k+j}$ or $D_{k \rightarrow k - j}$ to denote $j$ consecutive invocations of the up (or down) operator.
\end{remark}

\begin{definition}
    Going forward, for integers $n$ and $0 \leq i \leq n$, we use $\mathcal{U}_i$ to denote the uniform distribution over $\binom{[n]}{i}$.
\end{definition}

Importantly, we have the following key property about the down-operator:

\begin{lemma}[KL-Contraction; See, for instance, Lemma 3.3 in \cite{Lee24}.]\label{lem:OneDownContraction}
    Let $0<k \leq n$ be integers. Let $\mu$ be an arbitrary distribution over $\binom{[n]}{k}$. Then,
    \[
    d_{\mathrm{KL}}(D_{k \rightarrow k-1}(\mu), \mathcal{U}_{k-1}) \leq \left ( 1 - \frac{1}{k}\right ) \cdot d_{\mathrm{KL}}(\mu, \mathcal{U}_{k}).
    \]
\end{lemma}

As an immediate corollary of \cref{lem:OneDownContraction}, we have the following:

\begin{corollary}\label{cor:ManyDownContraction}
    Let $0\leq t < k \leq n$ be integers. Let $\mu$ be an arbitrary distribution over $\binom{[n]}{k}$. Then,
    \[
    d_{\mathrm{KL}}(D_{k \rightarrow t}(\mu), \mathcal{U}_{t}) \leq \frac{t}{k} \cdot d_{\mathrm{KL}}(\mu, \mathcal{U}_{k}).
    \]
\end{corollary}

\begin{proof}
When $t = k-1$, the claim holds trivially by \cref{lem:OneDownContraction}. More generally, by induction we can see that
\begin{align*}
d_{\mathrm{KL}}(D_{k \rightarrow t}(\mu), \mathcal{U}_{t})  &= d_{\mathrm{KL}}( D_{t+1 \rightarrow t}D_{t+2 \rightarrow t+1}\dots D_{k \rightarrow k-1}(\mu), \mathcal{U}_{t})\\
&\leq \left ( 1 - \frac{1}{t+1}\right ) \cdot d_{\mathrm{KL}}(D_{t+2 \rightarrow t+1}\dots D_{k \rightarrow k-1}(\mu), \mathcal{U}_{t+1})\\
&\leq  \left ( 1 - \frac{1}{t+1}\right )  \cdot \frac{t+1}{k} \cdot d_{\mathrm{KL}}(\mu, \mathcal{U}_{k})\\
& =\frac{t}{k} \cdot d_{\mathrm{KL}}(\mu, \mathcal{U}_{k}).
\end{align*}
\end{proof}

We then have the following corollary which holds for \emph{up-operators}:

\begin{corollary}\label{cor:ManyUpContraction}
    Let $0\leq m < i \leq n$ be integers. Let $\mu$ be an arbitrary distribution over $\binom{[n]}{m}$. Then,
    \[
    d_{\mathrm{KL}}(U_{m \rightarrow i}(\mu), \mathcal{U}_{i}) \leq \frac{n-i}{n-m} \cdot  d_{\mathrm{KL}}(\mu, \mathcal{U}_{m}).
    \]
\end{corollary}

\begin{proof}
    Consider the operator $\mathrm{Complement}$ which takes a distribution over $\binom{[n]}{m}$ and maps it to a distribution over $\binom{[n]}{n - m}$, where $\mathrm{Complement}(\mu)(S) = \mu([n] - S)$. We can then observe that 
    \[
    U_{m \rightarrow i}(\mu) = \mathrm{Complement} (D_{n - m \rightarrow n - i} (\mathrm{Complement}(\mu))),
    \]
    and so 
    \[
    \mathrm{Complement}(U_{m \rightarrow i}(\mu)) = D_{n - m \rightarrow n - i} (\mathrm{Complement}(\mu)).
    \]
    Likewise, it is also the case that $d_{\mathrm{KL}}(\mu \Vert \mathcal{U}_{m}) = d_{\mathrm{KL}}(\mathrm{Complement}(\mu) \Vert \mathcal{U}_{n-m})$.

    Thus, we see that 
    \begin{align*}
        d_{\mathrm{KL}}(U_{m \rightarrow i}(\mu), \mathcal{U}_{i}) & = d_{\mathrm{KL}}(\mathrm{Complement}(U_{m \rightarrow i}(\mu)), \mathcal{U}_{n-i})\\
        & = d_{\mathrm{KL}}(D_{n - m \rightarrow n - i} (\mathrm{Complement}(\mu)), \mathcal{U}_{n-i}) \\
        & \leq \frac{n-i}{n-m} \cdot \dKL(\mathrm{Complement}(\mu), \mathcal{U}_{n-m})\\
        & = \frac{n-i}{n-m} \cdot \dKL(\mu, \mathcal{U}_{m}),
    \end{align*}
    where the only inequality in the above sequence uses \cref{cor:ManyDownContraction}.
\end{proof}

Finally, we have the following fact, which is itself simply a strengthened version of Pinsker's inequality (which itself would still suffice for our purposes).

\begin{fact}[See, for instance, Lemma 2.1 in \cite{liu2020second}.]\label{fact:TransferProbability}
    Let $\nu, \pi$ be distributions over a set $\mathcal{X}$, and suppose that $d_{\mathrm{KL}}(\nu \Vert  \pi) \leq K$. If an event $G$ satisfies $\nu(G) \geq 1/2$, then 
    \[
    \pi(G) \geq \frac{1}{4} \cdot 2^{-2K}.
    \]
\end{fact}
\begin{proof}[Proof due to \cite{liu2020second}]
By the data processing property of relative entropy, we have
\begin{align*}
\dKL(\nu \Vert \pi) & \geq \nu(G) \log \frac{\nu(G)}{\pi(G)}+(1-\nu(G))\log \frac{1-\nu(G)}{1-\pi(G)}\\
& \geq \nu(G)\log \frac{1}{\pi(G)}- h(\nu(G)),
\end{align*}
where $h(x)=-x\log x-(1-x)\log (1-x)$ is the binary entropy function. Since $\nu(G)\geq 1/2$ and $h(x)\leq 1$, we obtain
\[
K\geq \dKL(\nu \Vert \pi) \geq \frac{1}{2}\log \frac{1}{\pi(G)} - 1,
\]
which implies that $\pi(G)\geq \frac{1}{4}\cdot 2^{-2K}$.
\end{proof}

\subsection{Bounds for Hypergeometric Distributions}

In the course of proving our strengthened decomposition theorem, we will require bounds on the concentration of hypergeometric distributions.

\begin{fact} \label{fact:KLBound}
For $\alpha,\beta \in [0,1]$, we have $\dKL(\alpha||\beta)\leq \log (e) \frac{(\alpha-\beta)^2}{\beta(1-\beta)}$.
\end{fact}
\begin{proof}
Using the inequality $\log x\leq \log(e)(x-1)$, we obtain
\begin{align*}
\dKL(\alpha||\beta) &= \alpha \log \frac{\alpha}{\beta} + (1-\alpha ) \log \frac{1-\alpha }{1-\beta}\\
& \leq \log(e)\left(\alpha\left(\frac{\alpha}{\beta}-1\right) + (1-\alpha) \left(\frac{1-\alpha}{1-\beta}-1\right)\right)\\
& = \log(e)\frac{(\alpha-\beta)^2}{\beta(1-\beta)}.
\end{align*}
\end{proof}

\begin{lemma} \label{lem:hyp}
Let $X\sim \Hyp(n,m,k)$ be a hypergeometric random variable with $m/n\leq 1/2$ and $k/n\leq 1/2$. Let $\mu=km/n$ be its mean. For any integer $x$ satisfying $|x-\mu|\leq \sqrt{\gamma \cdot \mu}$ and $0\leq x\leq \min(m,k)$, we have
\[
\Pr[X=x] \geq \frac{1}{(n+1)^2} 2^{-8\gamma}.
\]
\end{lemma}
\begin{proof}
We have
\[
\Pr[X=x] = \frac{{m\choose x}{n-m\choose k-x}}{{n\choose k}} = \frac{{k\choose x}{n-k \choose m-x}}{{n\choose m}}.
\]
Using the entropy bound (see, for instance, Example 11.1.3 of \cite{CT06})
\[
\frac{1}{n+1} 2^{nH(k/n)} \leq {n\choose k}\leq 2^{nH(k/n)},
\]
where $H(x)= -x\log x-(1-x)\log (1-x)$ is the binary entropy function, we obtain
\begin{align*}
\Pr[X=x]& \geq \frac{1}{(k+1)(n-k+1)} 2^{kH(x/k)+(n-k)H((m-x)/(n-k))-nH(m/n)}\\
& \geq \frac{1}{(n+1)^2} 2^{n(qH(\alpha)+(1-q)H(\beta)-H(p)).}
\end{align*}
where we define $\alpha=x/k, \beta=(m-x)/(n-k)$ and $p=m/n,q=k/n$.

Now we bound the exponent. We first rewrite it using KL divergence. Since $q\alpha+(1-q)\beta=p$, 
\[
qH(\alpha)+(1-q)H(\beta)-H(p)= -q\cdot \dKL(\alpha \Vert p) - (1-q) \cdot \dKL(\beta \Vert p).
\]
Applying \cref{fact:KLBound}, we obtain
\[
q\cdot \dKL(\alpha \Vert p) + (1-q) \cdot \dKL(\beta \Vert p) \leq \log(e)\frac{q(\alpha-p)^2+(1-q)(\beta-p)^2}{p(1-p)}.
\]
Since
\[
\alpha-p = \frac{x}{k}-\frac{m}{n} = \frac{x-\mu}{nq}, \quad \beta-p=\frac{m-x}{m-k}-\frac{m}{n}=\frac{\mu-x}{n(1-q)},
\]
we have
\[
q(\alpha-p)^2+(1-q)(\beta-p)^2  = \frac{(x-\mu)^2}{n^2q(1-q)} \leq \frac{\gamma\cdot \mu}{n^2q(1-q)}.
\]
Therefore,
\begin{align*}
q\cdot \dKL(\alpha \Vert p) + (1-q) \cdot \dKL(\beta \Vert p)& \leq \log(e)\frac{\gamma\cdot \mu}{n^2p(1-p)q(1-q)}\\
& = \log(e)\frac{\gamma}{n(1-p)(1-q)} && \text{(since $\mu=npq$)}\\
& \leq \log(e)\frac{4\gamma}{n} &&  \text{(since $p,q\leq 1/2$)}\\
& \leq \frac{8\gamma}{n}.
\end{align*}
Substituting this into the previous bound yields
\[
\Pr[X=x] \geq \frac{1}{(n+1)^2} 2^{n\cdot (-8\gamma/n)} = \frac{1}{(n+1)^2} 2^{-8\gamma}.
\]
\end{proof}

\section{Density Preservation When Subsampling Dense Sets}\label{sec:densityPreservation}

Throughout this section, we fix a universe of $n$ elements. We likewise fix an integer $i \leq n/2$, and set $\mathcal{F} \subseteq \binom{[n]}{i}$, and will assume that $\mathcal{F}$ is reasonably dense, namely that $\frac{|\mathcal{F} \cap \binom{[n]}{i}|}{\binom{n}{i}} = p $ is not negligible. 

Now, we will subsample a set $S$ of exactly $m = n \cdot (1 - 1 / \log(n))$ many elements from $[n]$; our goal is to show that after doing this subsampling, it is still the case that $\frac{\mathcal{F} \cap \binom{S}{i}}{\binom{m}{i}}$ is large. Intuitively, this is showing that the \emph{relative density} of this set $\mathcal{F}$ does not decrease too much when we subsample our universe.

We state this formally below:

\begin{theorem}\label{thm:preserveDensitySubsample}
    Let $n,m,i$ be integers with $i\leq m\leq n$. Let $\mathcal{F} \subseteq \binom{[n]}{i}$ and define 
    \[
    p = \frac{|\mathcal{F}|}{\binom{n}{i}}.
    \]
    Now, sample a set $S \subseteq [n]$ with exactly $m$ elements. Then, 
    \[
    \Pr_{S} \left [ \frac{\left|\mathcal{F} \cap \binom{S}{i} \right|}{\binom{m}{i}} \geq \frac{p}{2}\right ] \geq \frac{1}{4} \cdot 2^{-2\frac{n-m}{n-i} \log(1/p)}.
    \]
\end{theorem}

\begin{proof}
    We let $\mu$ be the uniform distribution over $\mathcal{F}$. We can then immediately see that 
    \[
    d_{\mathrm{KL}}(\mu, \mathcal{U}_i) = \sum_{T \in \mathcal{F}} \mu(T) \cdot \log \left ( \frac{\mu(T)}{\mathcal{U}_i(T)}\right ) = |\mathcal{F}| \cdot \frac{1}{|\mathcal{F}|} \cdot \log \left ( \frac{\binom{n}{i}}{|\mathcal{F}|}\right ) = \log(1 /p).
    \]

    Now, let $\nu = U_{i \rightarrow m}(\mu)$. By \cref{cor:ManyUpContraction}, we then see that 
    \begin{align}\label{eq:KLUB}
    d_{\mathrm{KL}}(\nu \Vert \mathcal{U}_m) \leq \frac{n-m}{n-i} \cdot d_{\mathrm{KL}}(\mu \Vert \mathcal{U}_i) \leq \frac{n-m}{n-i} \cdot \log(1/p).
    \end{align}

    Next, define $\mathcal{G} = \left \{S \in\binom{[n]}{m}: \frac{|\mathcal{F} \cap \binom{S}{i}|}{\binom{m}{i}} \geq \frac{p}{2} \right \}$. We now claim that $\nu(\bar{\mathcal{G}}) \leq 1/2$. Indeed, to see this, we observe that 
    \[
    \nu(\bar{\mathcal{G}}) = \sum_{S \in \bar{\mathcal{G}}} \nu(S) = \sum_{S \in \bar{\mathcal{G}}} \frac{1}{\binom{n}{m}} \cdot \binom{n}{m} \cdot \nu(S).
    \]
    Note that $\nu(S)$ is the probability that, starting with a random set $T \in \mathcal{F}$, adding a uniformly random set of $m-i$ elements yields the set $S$. This probability is exactly
    \begin{align*}
        \nu(S) & = \Pr_{T \sim \mu} [T \subseteq S] \cdot \Pr_{A\sim {[n]-T\choose m-i}}[A=S-T]\\
        & =  \frac{\left |\mathcal{F}\cap \binom{S}{i} \right |}{|\mathcal{F}|} \cdot \frac{1}{\binom{n-i}{m-i}}.
    \end{align*}
    In particular, we then see that 
    \[
    \binom{n}{m} \cdot \nu(S) = \frac{\left |\mathcal{F}\cap \binom{S}{i}\right |}{|\mathcal{F}|} \cdot \frac{\binom{n}{m}}{\binom{n-i}{m-i}} = \frac{\left |\mathcal{F}\cap \binom{S}{i}\right |}{|\mathcal{F}|} \cdot \frac{\binom{n}{i}}{\binom{m}{i}} =\frac{\left |\mathcal{F}\cap \binom{S}{i}\right |}{\binom{m}{i}} \cdot \frac{1}{p}.
    \]

    Using the above, we then see that
    \begin{align*}
        \nu(\bar{\mathcal{G}}) & = \sum_{S \in \bar{\mathcal{G}}} \frac{1}{\binom{n}{m}} \cdot \binom{n}{m} \cdot \nu(S)\\
        & = \sum_{S \in \bar{\mathcal{G}}} \frac{1}{\binom{n}{m}} \cdot \frac{\left |\mathcal{F}\cap \binom{S}{i}\right |}{\binom{m}{i}} \cdot \frac{1}{p} \\
        & \leq \sum_{S \in \bar{\mathcal{G}}} \frac{1}{\binom{n}{m}} \cdot \frac{1}{2} \leq \frac{1}{2}.
    \end{align*}
    where we have used that, by definition, every set $S \in \bar{\mathcal{G}}$ satisfies $\frac{|\mathcal{F} \cap \binom{S}{i}|}{\binom{m}{i}} < \frac{p}{2}$.

    To conclude, we have $\nu(\mathcal{G}) \geq 1/2$. Along with \cref{eq:KLUB} and \cref{fact:TransferProbability}, this then implies that \[
    \mathcal{U}_m(\mathcal{G}) \geq \frac{1}{4} \cdot 2^{-2\frac{n-m}{n-i} \log(1/p)}.
    \]
    This exactly means that, with probability at least $\frac{1}{4} \cdot 2^{-2\frac{n-m}{n-i} \log(1/p)}$ over the sampled set $S$ of $m$ elements, it is the case that $\frac{|\mathcal{F} \cap \binom{S}{i}|}{\binom{m}{i}} \geq \frac{p}{2}$, as we desire.
\end{proof}

\section{Decomposition Algorithm}\label{sec:decomposition}

In this section, we revisit the matroid decomposition algorithm of \cite{KPS25} and obtain a slight improvement in the round complexity. The presentation largely follows the structure of Section 4 of \cite{KPS25}.

\subsection{Finding Sets with Many Circuits}

We start by defining the greedily-optimal sets. Recall that for a matroid $\M=(E,\I)$ and a subset $S\subseteq E$, we let
\[
q_S= \Pr_\pi[C_\pi\subseteq S].
\]
\cite{KPS25} shows that by sampling $\text{poly}(n)$ random permutations, one can estimate $q_S$ to small error for every $S$. In particular, let $\widehat{q}_S$ denote the empirical estimate of $q_S$ obtained from $n^{15}$ random permutations. A simple application of a Chernoff bound yields the following statement.

\begin{claim} \label{clm:boundedErrorEstimation}
For a matroid $\M=(E,\I)$ and every subset $S\subseteq E$,
\[
\left|\widehat q_S- q_S\right|\leq \frac{1}{n^{10}}
\]
with probability $1-2^{-n}$.
\end{claim}

\begin{definition}
For a matroid $\M=(E,\I)$, we say that a set $S\subseteq E$ is \emph{greedily-optimal} if 
\[
    \widehat{q}(S) \geq 1-\frac{1}{n^{7}}+\frac{|S|}{n^{8}},
\]
and there is no element $x \in S$ such that 
\[
    \widehat{q}_{S-x} \geq 1-\frac{1}{n^{7}}+\frac{|S|-1}{n^{8}},
\]
\end{definition}

There is a simple algorithm for finding greedily-optimal sets: we start with the ground set $E$ and continue to delete elements until the condition no longer holds.

\begin{algorithm}[H]
\caption{FindGreedilyOptimal$(\M=(E,\I))$}\label{alg:FindGreedilyOptimal}
Multiset $\mathcal{C}\gets \emptyset$\\
\For {$i\in[n^{15}]$ in parallel} {
    Draw a random permutation $\pi$.\\
    $\mathcal{C}\gets \mathcal{C}\cup \{C_\pi\}$.
}
$S \gets E$. \\
\While{True}{
\For{$x \in S$}{
$\widehat{q}_{S - x} 
\gets\frac{ \left |C\in \mathcal{C}: C \subseteq S - \{x \} \right |}{|\mathcal{C}|}$
}

\If{$\exists x \in S$ s.t. $ \widehat{q}_{S - x} \geq 1-\frac{1}{n^{7}} + \frac{|S|-1}{n^{8}}$} {
Let $x$ be the first such element\\
$S \leftarrow S - x $
} 
\Else {
\Return{$S$}
}
}
\end{algorithm}

\begin{claim} \label{clm:peelTime}
\cref{alg:FindGreedilyOptimal} runs in one round and finds a greedily-optimal set $S$ in $\M$. 
\end{claim}

We now establish the following claim, which shos that every element in the greedily-optimal set has non-trivial marginal probability.

\begin{claim}\label{clm:highProbCircuit}
Let $\M=(E,\I)$ be a matroid, and let $S$ be a greedily-optimal set. Then, with probability $\geq 1 - 2^{-n}$, for every $x \in S$,
\[
    p_{x,\M|_{S_i}} \geq \frac{1}{n^{9}}.
\]
\end{claim}
\begin{proof}
    By the definition of being greedily-optimal, for every element $x \in S$, it must be that 
    \[
    \widehat{q}_{S - x} < 1- \frac{1}{n^{7}} +\frac{|S|-1}{n^{8}},
    \]
    and that 
    \[
    \widehat{q}_{S} \geq 1-\frac{1}{n^{7}}+\frac{|S|}{n^{8}}.
    \]
    In particular, this means
    \[
    \widehat{q}_{S}  - \widehat{q}_{S - x} \geq \frac{1}{n^{8}}.
    \]
    By \cref{clm:boundedErrorEstimation}, with high probability, we have
    \[
    q_S - q_{S - x}\geq \frac{1}{n^{9}}.
    \]

    Now,for a matroid $\M=(E,\I)$ and $x\in T, T\subseteq E$, we define the auxiliary value $p_{x,T,\M}$ as
    \[
    p_{x,T,\M} = \Pr_{\pi\sim \M}[\{x\}\subseteq C_\pi\subseteq T].
    \]
    We can observe that $p_{x,\M|_S} = p_{x,S,\M|_S} \geq p_{x,S,\M}$. The inequality is because whenever we sample in accordance to a permutation $\pi$ over $E$, and recover a circuit $C_{\pi}$ such that $x \in C_{\pi}$ and $C_{\pi} \subseteq S$, the same permutation, if restricted to $S$ and used to sample elements of $S$, would have given a circuit such that $x \in S$.

    Finally, we can observe that $p_{x,S,\M} = q_{S} - q_{S - x}$, as $q_{S} - q_{S - x}$ is exactly the probability that a circuit, when sampled from $\M$, is contained in $S$ and uses the element $x$. Thus, we conclude that $p_{x,\M|_S} \geq p_{x,S,\M} = q_{S} - q_{S - x}\geq n^{-9}$.

    The high probability bound follows from \cref{clm:boundedErrorEstimation}.
\end{proof}

In addition, given a greedily-optimal set $S$, one can recover an independent set of size $\Omega\left(\frac{\alpha(S)}{|S|} \cdot n\right)$ via random sampling.

\begin{claim}[Claim 5.1 of \cite{KPS25}] \label{clm:largeAlpha}
Let $\M=(E,\I)$ be a matroid on $n$ elements, and let $S$ be a greedily-optimal set. Then, with probability $\geq 1 - 2^{-n}$, for $\ell = \frac{\alpha(S)}{10|S|} n$, we have
\[
\Pr_\pi[\Ind(\{\pi(1),\dots,\pi(\ell)\})=1]\geq \frac{1}{4}.
\]
\end{claim}

\subsection{Iterative Matroid Decomposition}
Now we are ready to describe the matroid decomposition algorithm. It first applies the procedure $\text{RemoveSmallCircuits}(\M)$ (see Algorithm 7 of \cite{KPS25}). This procedure runs in one round and simply eliminates all circuits of size at most $50$ from the matroid; it is introduced for technical reasons in the probabilistic argument. The algorithm then proceeds by iteratively peeling off greedily-optimal sets from the remaining matroid using the algorithm established previously.

\begin{algorithm}[H]
\caption{IterativePeeling$(\M=(E,\I))$}\label{alg:iterativePeel}
    $\M\gets \text{RemoveSmallCircuits}(\M)$ \tcp{Algorithm 7 in \cite{KPS25}}
    $n \gets |E|,k \gets 0$ \\
    \While{$|\M|\geq n/2$}{
    $k\gets k +1$\\
    $S_k \leftarrow \mathrm{FindGreedilyOptimal}(\M)$. \tcp{\cref{alg:FindGreedilyOptimal}}
    \If{$\alpha(S_k)/|S_k|> \frac{1}{100 \log(n)}$ or $|S_k|>n/4$} {
        \Return{$S_1,\dots ,S_{k}$}
    }
    $\M \leftarrow \M \setminus S_{k}$.
    }
    \Return{$S_1, \dots ,S_{k}$}
\end{algorithm}

Let $S_1,\dots,S_k$ be the sets produced by the above algorithm. We first observe that because the procedure $\text{RemoveSmallCircuits}(\M)$ eliminates all circuits of size $\leq 50$ at the beginning, we always have $\alpha(S_i)\geq 50$. In addition, by the stopping condition, we have $\alpha(S_i)/|S_i|\leq 1/2$ and $|S_i|\leq n/4$ for every $i<k$.

We provide the following characterization of how the $\alpha$-value evolve throughout the process, which is essentially a strengthening of Claim 4.13 in \cite{KPS25}:

\begin{claim}\label{clm:peelingSetsAlpha}
Let $\mathcal{M}$ be a matroid, and let $S_1, \dots, S_k$ be a sequence of sets that are peeled off in accordance with \cref{alg:iterativePeel}. Then, with probability $\geq 1 - k \cdot 2^{-n}$, for any $i<j<k$ where $|S_i|\leq 2|S_j|$, we have
\[
\alpha(S_j) \geq \frac{\alpha(S_i)|S_j|}{|S_i|} + \sqrt{\frac{\alpha(S_i)|S_j|}{|S_i|} \cdot \frac{\log n}{16}}.
\]    
\end{claim}

\begin{proof}
The high probability bound is necessary because we assume that \cref{clm:boundedErrorEstimation} holds for each set $S_i$ that is peeled off. 

Let $m$ be the size of the matroid $\M$ at the start of the $i$th iteration. Then we have $m\geq n/2$ and $|S_i|/m\leq 1/2,|S_j|/m\leq 1/2$. Let $\ell = \frac{\alpha(S_i)}{|S_i|} m$. Let $U$ be a uniformly random subset of $E$ of cardinality $\ell$, we define random variables $X_i = |U\cap S_i|, X_j =|U \cap S_j|$. Note that $X_i\sim \Hyp(n,|S_i|,\ell),X_j\sim \Hyp(n,|S_j|,\ell)$.

For the sake of contradiction, suppose
\[
\alpha(S_j) \leq \frac{\alpha(S_i)|S_j|}{|S_i|} +  \sqrt{\frac{\alpha(S_i)|S_j|}{|S_i|} \cdot \frac{\log n}{16}}.
\]
We will show that
\[
\Pr[X_i< \alpha(S_i) \land X_j\geq \alpha(S_j)] \geq \frac{4}{n^{6}}.
\]
Consider the $i$th iteration of \cref{alg:iterativePeel}, the above implies that for a random permutation $\pi$, there is a $\geq 4/n^{6}$ probability that there are less than $\alpha(S_i)$ elements from $S_i$ and more than $\alpha(S_j)$ elements from $S_j$ in the first $\ell$ elements of $\pi$. Conditioned on this event, with probability at least $1/4$, there is no circuit within $S_i$ but there is a circuit in $S_j$. This implies that with probability $\geq1/n^{6}$, the first circuit appears outside of $S_i$ when adding elements according to the order of a random permutation $\pi$. But on the other hand, since $S_i$ is a greedily-optimal set in the $i$th iteration, by \cref{clm:boundedErrorEstimation}, we have $q_{S_i} \geq 1-n^{-7}-n^{-10}> 1-n^{-6}$. This is a contradiction.

We first observe that the events $X_i< \alpha(S_i)$ and $X_j\geq \alpha(S_j)$ are positively correlated. Thus,
\[
\Pr[X_i< \alpha(S_i) \land X_j\geq \alpha(S_j)] \geq \Pr[X_i< \alpha(S_i)] \cdot \Pr[X_j\geq \alpha(S_j)].
\]
We analyze the two in turn in the following.

Since $X_i\sim \Hyp(m,|S_i|,\ell)$, we have $\E[X_i]=\ell \frac{|S_i|}{m} = \alpha(S_i)$. Thus,
\[
\Pr[X_i< \alpha(S_i)] \geq \Pr[X_i=\E[X_i]-1].
\]
Let $x=\E[X_i]-1$. Since $\E[X_i]=\alpha(S_i)\geq 50$ by our assumption, we have
\[
\left|x-\E[X_i]\right| = 1 \leq \sqrt{\E[X_i]}.
\]
Since $|S_i|/m\leq 1/2$ and $\ell/n=\alpha(S_i)/|S_i|\leq 1/2$, we can invoke \cref{lem:hyp} and obtain
\[
\Pr[X_i< \alpha(S_i)] \geq \Pr[X_i=x] \geq \frac{1}{(m+1)^2} 2^{-4} \geq  \frac{1}{16(n+1)^2}.
\]

Similarly, since $X_j\sim \Hyp(m,|S_j|,\ell)$, we have $\E[X_j] = \ell \frac{|S_j|}{m} = \frac{\alpha(S_i)|S_j|}{|S_i|}$. Thus,
\begin{align*}
\Pr[X_j\geq \alpha(S_j)] &\geq \Pr\left[X_j\geq \frac{\alpha(S_i)|S_j|}{|S_i|} +  \sqrt{\frac{\alpha(S_i)|S_j|}{|S_i|} \cdot \frac{\log n}{16}}\right]\\
& = \Pr\left[X_j\geq \E[X_j]+\sqrt{\E[X_j]\cdot \frac{\log n}{16}}\right]\\
& \geq \Pr\left[X_j = \left\lceil \E[X_j]+\sqrt{\E[X_j]\cdot \frac{\log n}{16}}\right\rceil\right].
\end{align*}
Let $x=\left\lceil \E[X_j]+\sqrt{\E[X_j]\log n/16}\right\rceil$. Since $\E[X_j]=\frac{\alpha(S_i)|S_j|}{|S_i|}\geq 25$ by our assumptions that $\alpha(S_i)\geq 50, |S_j|/|S_i|\geq 1/2$, we have
\[
|x-\E[X_j]| \leq \sqrt{\E[X_j]\log n/16} + 1 \leq \sqrt{\E[X_j]\log n/8}.
\]
Since $|S_j|/m\leq 1/2$ and $\ell/m=\alpha(S_i)/|S_i|\leq 1/2$, we can invoke \cref{lem:hyp} and obtain
\[
\Pr[X_j\geq \alpha(S_j)] \geq \Pr[X_j=x] \geq \frac{1}{(m+1)^2} 2^{-\log n} \geq \frac{1}{n(n+1)^2}.
\]

We conclude that
\begin{align*}
\Pr[X_i< \alpha(S_i) \land X_j\geq \alpha(S_j)] &\geq \Pr[X_i< \alpha(S_i)] \cdot \Pr[X_j\geq \alpha(S_j)]\\
& \geq \frac{1}{16(n+1)^2}\cdot \frac{1}{n(n+1)^2}\\
&\geq \frac{4}{n^6},
\end{align*}
as desired.
\end{proof}

We now show that, since the $\alpha$-value grows at a sufficiently fast rate, the number of sets peeled off throughout the process is bounded.

\begin{lemma}\label{lem:quadraticGrowth}
Let $\mathcal{M}$ be a matroid, and let $S_1, \dots, S_k$ be a sequence of sets that are peeled off in accordance with \cref{alg:iterativePeel}. Then, with probability $\geq 1 - k \cdot 2^{-n}$ over the decomposition algorithm, for any $\ell\in [\log n]$ an integer, if one defines $T = \{ i \in [k]: |S_i| \in [2^{\ell}, 2^{\ell+1} -1]\}$ and $\gamma = |T|$, and let $a_1, \dots ,a_{\gamma}$ denote the indices in $T$, then, 
\[
    \gamma = O\left(\sqrt{2^{\ell}} /\log n\right).
\]
\end{lemma}
\begin{proof}
As $|S_{a_i}|\leq 2|S_{a_{i+1}}|$, by \cref{clm:peelingSetsAlpha}, we have
\[
\alpha(S_{a_{i+1}}) \geq \frac{\alpha(S_{a_i})|S_{a_{i+1}}|}{|S_{a_i}|} + \sqrt{\frac{\alpha(S_{a_i})|S_{a_{i+1}}|}{|S_{a_i}|} \cdot \frac{\log n}{16}}.
\]
Now, we multiply both sides with $2^{\ell}/|S_{a_{i+1}}|$:

\begin{align*}
\frac{\alpha(S_{a_{i+1}})}{|S_{a_{i+1}}|} \cdot 2^{\ell} & \geq \frac{\alpha(S_{a_i})}{|S_{a_i}|} \cdot 2^{\ell} + \sqrt{\frac{\alpha(S_{a_i})}{|S_{a_i}|} \cdot \frac{2^{\ell}}{|S_{a_{i+1}}|} \cdot 2^{\ell} \cdot \frac{\log n}{16}} \\
& \geq \frac{\alpha(S_{a_i})}{|S_{a_i}|} \cdot 2^{\ell} + \sqrt{\frac{\log n}{32}} \cdot  \sqrt{\frac{\alpha(S_{a_i})}{|S_{a_i}|} \cdot 2^{\ell}}
\end{align*}
If we set $X_i = \frac{\alpha(S_{a_{i}})}{|S_{a_{i}}|}\cdot 2^{\ell}$. We obtain the recurrence that 
\[
X_{i+1} \geq X_i + \sqrt{\frac{X_i\log n}{32}}.
\]
As $X_1\geq 1$, the recurrence implies that $X_\gamma=\Omega(\gamma^2\log n)$.  Since $\alpha(S_{a_\gamma}) \leq \frac{|S_{a_\gamma}|}{100 \log(n)}\leq \frac{2^{\ell+1}}{100 \log(n)}$, we have
\[
\gamma = O\left(\sqrt{2^{\ell}}/\log n\right).
\]

The probability bound holds by assuming that \cref{clm:boundedErrorEstimation} holds for each peeled set in the decomposition.
\end{proof}

\begin{lemma}
\label{lem:peel-round}
Let $\mathcal{M}$ be a matroid, and let $S_1, \dots ,S_k$ be a sequence of sets that are peeled off in accordance with \cref{alg:iterativePeel}. Then, with probability $\geq 1 - k \cdot 2^{-n}$, we have $k=O(n^{1/3}/\log ^{2/3}n)$.
\end{lemma}
\begin{proof}
For every $\ell\in [\log n]$, we let $T_\ell= \{ i \in [k]: |S_i| \in [2^{\ell}, 2^{\ell+1} -1]\}$. It follows from \cref{lem:quadraticGrowth} (and this is the source of the probability bound as well) that $|T_{\ell}| = O\left(\sqrt{2^{\ell}}/\log n\right)$.  Let $\ell^*$ be the largest integer such that $2^{\ell^*}\leq n^{2/3} \log ^{2/3} n$. We have
\begin{align*}
k& = \sum_{i\in [k]} \textbf{1}[|S_i|\leq n^{2/3} \log ^{2/3} n] + \sum_{i\in [k]} \textbf{1}[|S_i|> n^{2/3} \log ^{2/3} n]\\
& \leq \sum_{i\in [k]} \textbf{1}[|S_i|\leq n^{2/3} \log ^{2/3} n] + \frac{n^{1/3}}{\log ^{2/3}n}\\
& \leq \sum_{\ell\leq \ell^*} |T_\ell|+ \frac{n^{1/3}}{\log ^{2/3}n}\\
& = \sum_{\ell\leq \ell^*} O\left(\sqrt{2^{\ell}}/\log n\right) + \frac{n^{1/3}}{\log ^{2/3}n}\\
& = O\left(\frac{n^{1/3}}{\log ^{2/3}n}\right).
\end{align*}
\end{proof}

Combining \cref{clm:peelTime}, \cref{clm:highProbCircuit}, \cref{clm:largeAlpha} and \cref{lem:peel-round}, the following lemma summarizes the properties satisfied by the decomposition algorithm.

\begin{lemma} \label{lem:decomposition}
    Let $\mathcal{M}$ be a matroid, and let $S_1,\dots,S_k$ be the sets peeled by the iterative decomposition algorithm described in \cref{alg:iterativePeel}. Then the following statements hold with probability $\geq 1 - 2^{-n/2}$:
    \begin{enumerate}
        \item $k = O(n^{1/3} / \log^{2/3}(n))$.
        \item The algorithm runs in $O(k)$ rounds and makes polynomially many queries.
        \item For each $S_i$, every element $x \in S_i$ satisfies $p_{x, \M|_{S_i}} \geq n^{-9}$.
        \item In the $i$th round of the decomposition, when the algorithm recovers $S_i$, for \[
        \ell = \frac{\alpha(S_i)}{10|S_i|} \cdot |\M - S_1 - \dots -S_{i-1}|,
        \]
        a random sample of $\ell$ elements from $\M - S_1 - \dots- S_{i-1}$ is independent with probability $\geq 1/4$. \label{item:findBigIS}
    \end{enumerate}
\end{lemma}

\section{A Near-Optimal Procedure for Deleting Redundant Elements}\label{sec:deletion}

\subsection{The Deletion Algorithm}

In this section, we focus on a single greedily-optimal set $S_i$ that is returned by the decomposition algorithm, with the properties that
\begin{enumerate}
    \item $\frac{\alpha(S_i)}{|S_i|} \leq \frac{1}{100 \log(n)}$;
    \item every element $x\in S_i$ satisfies $p_{x,\M|_{S_i}}\geq n^{-9}$.
\end{enumerate}

\cite{KPS25} shows that these properties imply that the set $S_i$ is rank deficient; that is, it contains many redundant elements. Thus, our goal is to design an efficient algorithm that recovers a set $R_i\subseteq S_i$ of redundant elements; namely, for every $x\in R_i$, we can certify that $x\in \mathrm{span}(S_i\setminus R_i)$. While the algorithms of \cite{KPS25,KPS26} were unable to recover a set of size $\tilde \Omega(|S_i|)$, we design an efficient algorithm that successfully recovers $\Omega(|S_i|/\log (n))$ elements.

Now we describe our algorithm for recovering redundant elements (\cref{alg:RecoverRedundantElements}). As discussed above, the input is a set $S$ such that $\frac{\alpha(S)}{|S|} \leq \frac{1}{100 \log(n)}$ and for every element $x \in S$, $p_{x, \M|_{S}} \geq n^{-9}$. The algorithm begins by drawing a random subset $S'\subseteq S$ of size $m=\lceil(1-1/\log n)|S|\rceil$. Next, it performs $n^{15}$ independent experiments in parallel. In each experiment, it draws a random permutation $\pi$ of $S'$, and computes the maximal independent prefix $P_\pi$ of $\pi$. We then mark every element in $S-S'$ that is spanned by $P_{\pi}$ as redundant. Finally, the algorithm returns all elements of $S-S'$ that were marked redundant.

\begin{algorithm}
    \caption{RecoverRedundantElements$(S)$}
    \label{alg:RecoverRedundantElements}
    $R\gets \emptyset$.\\
    Draw a random subset $S'\subseteq S$ of size $m=\lceil(1-1/\log n) |S|\rceil$.\\
    \For{$\ell \in [n^{15}]$ in parallel} {
        Draw a random permutation $\pi:[m]\to S'$.\\
        \For{$i\in [m]$ in parallel} {
            Query $\Ind(\{\pi(1),\dots,\pi(i)\})$.
        }
        Let $t$ be the smallest index such that for $P_\pi=\{\pi(1),\dots,\pi(t-1)\}$, $\Ind(P_\pi)=1$ and $\Ind(P_\pi+\pi(t))=0$.\\
        \For{$x\in S- S'$ in parallel} {
            Query $\Ind(P_\pi+\{x\})$.\\
            \If{$\Ind(P_\pi+\{x\})=0$}{
                $R\gets R\cup \{x\}$.
            }
        }
    }
    \Return $R$
\end{algorithm}

The following observations follow immediately from the description of the algorithm.

\begin{observation}
By the end of \cref{alg:RecoverRedundantElements}, we have $R\subseteq S\setminus S'$. Moreover, every element $x\in R$ satisfies $x\in \mathrm{span}(S')$.
\end{observation}

\begin{observation}
\cref{alg:RecoverRedundantElements} can be implemented in 1 round using polynomially many independence queries.
\end{observation}
\begin{proof}
For a permutation $\pi$, it suffices to query $\Ind(\{\pi(1),\dots,\pi(i)\})$ for every $i\in[m]$, and $\Ind(\{\pi(1),\dots,\pi(i),x\})$ for every $i\in[m]$ and $x\in S- S'$.
\end{proof}

The remainder of this section is devoted to proving the following theorem.

\begin{theorem} \label{thm:redundantElements}
Let $S$ be a set such that $\frac{\alpha(S)}{|S|} \leq \frac{1}{100 \log(n)}$ and for every element $x \in S$, $p_{x, \M|_{S}} \geq n^{-9}$. Now, consider running $n^2$ independent trials of \cref{alg:RecoverRedundantElements}. Then, with probability $1 - 2^{-n}$, at least one trial returns a set $R$ of redundant elements of size $|R|=\Omega \left ( \frac{|S|}{\log (n)} \right)$.
\end{theorem}

\subsection{Analysis of the Deletion Algorithm}

Towards analyzing our deletion algorithm, we have the following claim:

\begin{claim}\label{clm:goodSetHighDensity}
Let $S$ be a set such that $\frac{\alpha(S)}{|S|} \leq \frac{1}{100 \log(n)}$ and for every element $x \in S$, $p_{x, \M|_{S}} \geq n^{-9}$. Then, for any element $x \in S$, there is a choice of a value $g_x \leq |S|/4$ such that
\[
\frac{\left | \left \{T \in \binom{S - \{x \}}{g_x}: \mathrm{Ind}(T) = 1, x \in \mathrm{span}(T) \right \} \right |}{\binom{|S| -1}{g_x}} \geq n^{-10}.
\]
\end{claim}

\begin{proof}
    First, by our hypothesis, we know that when we sample a random permutation $\pi$ over $S$, the probability that element $x$ participates in the first circuit $C_{\pi}$ is at least $n^{-9}$.
    
    At the same time, we can observe that because $\frac{\alpha(S)}{|S|} \leq \frac{1}{100 \log(n)}$, 
    \[
    \Pr_{\pi \sim S}[C_{\pi} \text{ forms before index }|S|/4] \geq 1 - (1/2)^{25 \log(n)} \geq 1 - \frac{1}{n^{25}}.
    \]
    This is because we define $\alpha(S)$ to be the median number of elements we need to sample to see the first circuit $C_{\pi}$ form; thus if we independently sample $\alpha(S)$ many elements $25 \log(n)$ times, the probability that a circuit \emph{hasn't} formed is $\leq (1/2)^{25 \log(n)}$.

    For any element $x \in S$, we can then see that
    \[
    \Pr_{\pi \sim S}[x \in C_{\pi} \wedge C_{\pi} \text{ forms before index }|S|/4] \geq \frac{1}{n^{9}} - \frac{1}{n^{25}} \geq \frac{1}{2n^{9}}.
    \]

    Now, by pigeonhole principle, we know that there must exist a choice of $0 \leq g_x < |S|/4$ such that 
    \[
    \Pr_{\pi\sim S}[x \in C_{\pi} \wedge C_{\pi} \text{ forms at index }g_x+1] \geq \frac{1/(2n^9)}{|S|/4} \geq n^{-10}.
    \]
    This implies that
    \[
    \frac{\left | \left \{T \in \binom{S}{g_x+1}: \mathrm{Ind}(T-\{x\}) = 1, x \in \mathrm{span}(T-\{x\}), x\in T \right \} \right |}{\binom{|S|}{g_x+1}} \geq n^{-10}.
    \]
    Consider $T'=T-\{x\}$, we have
    \[
    \frac{\left | \left \{T' \in \binom{S-\{x\}}{g_x}: \mathrm{Ind}(T') = 1, x \in \mathrm{span}(T') \right \} \right |}{\binom{|S|}{g_x+1}} \geq n^{-10}.
    \]
    Since ${|S|\choose g_x+1}\geq {|S|-1\choose g_x}$, we conclude that
    \[
    \frac{\left | \left \{T' \in \binom{S - \{x \}}{g_x}: \mathrm{Ind}(T') = 1, x \in \mathrm{span}(T') \right \} \right |}{\binom{|S| -1}{g_x}} \geq n^{-10}.
    \]
\end{proof}

Importantly, building on this claim, we now get the following lemma which shows that, even under mild subsampling, the density of ``good'' sets remains high:

\begin{lemma} \label{lem:goodSetAfterSubsample}
    Let $S$ be a set such that $\frac{\alpha(S)}{|S|} \leq \frac{1}{100 \log(n)}$ and for every element $x \in S$, $p_{x, \M|_{S}} \geq n^{-9}$. Now, fix an element $x \in S$, and sample a set $S' \subseteq S - \{x \}$ with $m=\lceil(1-1/\log n) |S|\rceil$ many elements. 
    
    Then, for any element $x \in S$, there is a choice of a value $g_x \leq |S|/4$ such that
\[
\Pr_{S'} \left [ \frac{\left | \left \{T \in \binom{S - \{x \}}{g_x}: \mathrm{Ind}(T) = 1, x \in \mathrm{span}(T) \right \} \cap \binom{S'}{g_x}\right |}{\binom{m}{g_x}} \geq \frac{1}{2n^{10}} \right ] = \Omega(1).
\]
\end{lemma}

\begin{proof}
    Let $S$ and $x \in S$ be given as such. For this element $x$, we then let $g_x \leq |S|/4$ be given as guaranteed by \cref{clm:goodSetHighDensity}. In particular, this guarantees that 
    \[
\frac{\left | \left \{T \in \binom{S - \{x \}}{g_x}: \mathrm{Ind}(T) = 1, x \in \mathrm{span}(T) \right \} \right |}{\binom{|S| -1}{g_x}} \geq n^{-10}.
\]

Now, in the language of \cref{thm:preserveDensitySubsample}, we let $\mathcal{F} = \left \{T \in \binom{S - \{x \}}{g_x}: \mathrm{Ind}(T) = 1, x \in \mathrm{span}(T) \right \}$, $m=\lceil(1-1/\log n) |S|\rceil$, the universe size is $|S| -1$, and let $i = g_x \leq |S|/4$. We likewise set \[
p = \frac{|\mathcal{F}|}{\binom{|S|-1}{i}} \geq n^{-10}.
\]
\cref{thm:preserveDensitySubsample} implies that when we subsample $S'$ from $S$ with exactly $m$ elements, that 
\[
\Pr_{S'} \left [ \frac{\left | \left \{T \in \binom{S - \{x \}}{g_x}: \mathrm{Ind}(T) = 1, x \in \mathrm{span}(T) \right \} \cap \binom{S'}{g_x}\right |}{\binom{m}{g_x}} \geq \frac{1}{2n^{10}} \right ] \geq \frac{1}{4} \cdot 2^{- 2\frac{
|S|-1-m}{|S|-1-i} \log(1/p)}.
\]
Since $m=\lceil(1-1/\log n) |S|\rceil$ and $i\leq |S|/4$, we have $\frac{|S|-1-m}{|S|-1-i} = O(1/\log n)$. Moreover, since $\log(1/p)=\log(n^{10})=O(\log n)$, the exponent above is $O(1)$. Therefore, the lower bound is at least $\Omega(1)$.
\end{proof}

Fix a subset $S'\subseteq S$ of size $m=\lceil(1-1/\log n) |S|\rceil$. We say an element $x\in S$ is \emph{good} with respect to $S'$ if $x\notin S'$ and
\[
\frac{\left | \left \{T \in \binom{S - \{x \}}{g_x}: \mathrm{Ind}(T) = 1, x \in \mathrm{span}(T) \right \} \cap \binom{S'}{g_x}\right |}{\binom{m}{g_x}} \geq \frac{1}{2n^{10}}.
\]

We show that, during the execution of \cref{alg:RecoverRedundantElements}, good elements are likely to be included in the set of redundant elements $R$.
\begin{claim} \label{clm:goodElementsSpanned}
Suppose $S'\subseteq S$ is the subset chosen in \cref{alg:RecoverRedundantElements}. Then every element $x$ that is good w.r.t. $S'$ is included in $R$ with probability at least $1-2^{-n^2}$.
\end{claim}
\begin{proof}
We show that
\[
\Pr_\pi [x\in \mathrm{span}(P_\pi)] \geq \frac{1}{2n^{10}}.
\]
The claim then follows since whenever $x\in \mathrm{span}(P_\pi)$, $x$ is added to $R$, and the probability that $x\notin \mathrm{span}(P_\pi)$ over all $n^{15}$ random permutations is at most
\[
\left(1-\frac{1}{2n^{10}}\right)^{n^{15}} \leq 2^{-n^2}.
\]

Let $T_\pi$ denote the set of first $g_x$ elements in the permutation $\pi$. Since $P_\pi$ is the maximal independent prefix of $\pi$, we have
\[
\Pr_\pi [x\in \mathrm{span}(P_\pi)]\geq \Pr_\pi[x\in \mathrm{span}(T_\pi)\mid \Ind(T_\pi)=1] \geq \Pr_\pi[x\in \mathrm{span}(T_\pi)\wedge \Ind(T_\pi)=1].
\]
Moreover,
\[
\Pr_\pi[x\in \mathrm{span}(T_\pi)\wedge \Ind(T_\pi)=1] = \Pr_{T\sim {S'\choose g_x}}[x\in \mathrm{span}(T)\wedge \Ind(T)=1].
\]
Since $x$ is good w.r.t. $S'$, we have $x\notin S'$, and the RHS is exactly equal to
\[
\frac{\left | \left \{T \in \binom{S - \{x \}}{g_x}: \mathrm{Ind}(T) = 1, x \in \mathrm{span}(T) \right \} \cap \binom{S'}{g_x}\right |}{\binom{m}{g_x}} \geq \frac{1}{2n^{10}}.
\]
This completes the proof.
\end{proof}

Next, we show that, in expectation over all subsets $S'\subseteq S$ of size $m$, there are many good elements w.r.t. to $S'$.

\begin{lemma}\label{lem:expectedNumberGoodElements}
    Let $S$ be a set such that $\frac{\alpha(S)}{|S|} \leq \frac{1}{100 \log(n)}$ and for every element $x \in S$, $p_{x, \M|_{S}} \geq n^{-9}$. Now, sample a subset $S' \subseteq S$ with $m=\lceil(1 - 1 / \log (n))|S| \rceil$ many elements. Let $X(S')$ be the random variable that denotes the number of good elements in $S$ w.r.t. $S'$. Then, there is an absolute constant $C > 0$ such that
    \[
    \E_{S'}[X(S')] \geq C\cdot |S- S'|.
    \]
\end{lemma}

\begin{proof}
By \cref{lem:goodSetAfterSubsample}, for an element $x\in S$,
\[
\Pr_{S'}[x\text{ is good w.r.t. }S'\mid x\not\in S'] \geq C
\]
for some constant $C$. Thus, we have
\[
\Pr_{S'}[x\text{ is good w.r.t. }S'] \geq C\cdot \frac{|S- S'|}{|S|}.
\]
Then, we have
\[
\E_{S'}[X(S')] \geq |S|\cdot C\cdot \frac{|S- S'|}{|S|}\geq C\cdot |S- S'|.
\]
\end{proof}

Now we are ready to prove \cref{thm:redundantElements}.

\begin{proof} [Proof of \cref{thm:redundantElements}]
By \cref{lem:expectedNumberGoodElements} and the fact that $X(S')\leq |S\setminus S'|$, Markov's inequality implies that
\[
\Pr_{S'}[X(S')>C/2\cdot |S\setminus S'|] \geq C'
\]
for some constant $C'$. By \cref{clm:goodElementsSpanned} together with a union bound over all elements, we have that in a single trial of \cref{alg:RecoverRedundantElements} all good elements w.r.t. $S'$ are included in $R$ with probability $1-n\cdot 2^{-n^2}$. Therefore, 
\[
\Pr_{S'}[|R|\geq C/2\cdot |S\setminus S'|] \geq C'/2.
\]

Consequently, if we run $n^2$ independent parallel instances of \cref{alg:RecoverRedundantElements}, then with probability at least $1 - (1 - C'/2)^{n^2} \geq 1 - 2^{-n}$, at least one instance returns a set $R$ of redundant elements of size $|R|\geq C/2 \cdot |S\setminus S'|=\Omega(|S|/\log (n))$. This concludes the theorem.
\end{proof}

\section{A Near-Optimal Basis Finding Algorithm}\label{sec:optalgorithm}

With the above established, we can now present the algorithm for finding a basis. The algorithm first applies the decomposition procedure to iteratively peel off sets from the matroid until either it encounters a set $S_k$ with $\alpha(S_k)/|S_k|\leq 1/100\log (n)$, or fewer than half of the elements remain in the matroid. In the former case, the algorithm finds an independent set of size $\Omega(n/\log (n))$ and contracts on it. In the latter case, the algorithm uses the new redundant elements recovery procedure to identify a set of $\Omega(n/\log (n))$ redundant elements, which are then deleted from the matroid.

We remark that for a given set $S$, we do not know the exact value of $\alpha(S)$. Nevertheless, we can approximate $\alpha(S)$ within a multiplicative factor of $2$ in a single round using polynomially many independence queries (see Definition 4.6 and Claim 4.8 of \cite{KPS25}), which suffices for our purposes.

\begin{algorithm}[H]
\caption{FindBasis$(\M=(E,\I))$}\label{alg:findBasis}
    $\M\gets \text{RemoveSmallCircuits}(\M)$ \tcp{Algorithm 7 in \cite{KPS25}}
    $n \gets |E|,k \gets 0$ \\
    \While{$\alpha(S_k)/|S_k|\leq 1/100\log (n), |S_k|\leq n/4$ and $|\M| \geq n/2$}{
    $k\gets k +1$\\
    $S_k \leftarrow \mathrm{FindGreedilyOptimal}(\M)$. \tcp{\cref{alg:FindGreedilyOptimal}}
    $\M \leftarrow \M \setminus S_{k}$. \\
    }
    \If{$\alpha(S_k)/|S_k|\geq 1/100\log (n)$}{
    Use \cref{lem:decomposition}, \cref{item:findBigIS} to find a set of at least
    \[
    \frac{\alpha(S_k)}{10|S_k|} \cdot \left | \M - S_1 - \dots - S_{k-1}\right | = \Omega(n/\log (n))
    \]
    independent elements. Denote this set by $R_{\mathrm{ind}}$. \\
    Recursively call $\text{FindBasis}(\M/R_{\mathrm{ind}})$.
    }
    \Else{
    \For{$i \in [k]$ in parallel}{
        \For{$\ell \in [n^2]$ in parallel}{
        Let $R_{i, \ell} = \mathrm{RecoverRedundantElements}(S_i)$.
        }
        Let $R_i = \mathrm{argmax}_{\ell \in [n^2]}|R_{i, \ell}|$.
    }
    Let $R_{\mathrm{redundant}} = \bigcup_{i \in [k]} R_i$. \\
    Recursively call $\text{FindBasis}(\M- R_{\mathrm{redundant}})$.
    }
\end{algorithm}

\begin{theorem}\label{thm:formalBasisFinding}
There is a randomized algorithm that, given any matroid $\M$ on $n$ elements, finds a basis of $\M$ with probability $\geq 2/3$ in $O(n^{1/3}\log^{1/3} (n))$ adaptive rounds, using only polynomially many independence queries per round.
\end{theorem}
\begin{proof}
We first show that \cref{alg:findBasis} correctly computes a basis of $\M$. Consider a single iteration of \cref{alg:findBasis}. In the first case, we find an independent set and contract it. In the second case, for every $i \in [k]$, we have $R_i \subseteq \mathrm{span}(S_i - R_i)$. Therefore, for $R_{\mathrm{redundant}} = \bigcup_{i \in [k]} R_i$, we have $R_{\mathrm{redundant}} \subseteq \mathrm{span}(E-R_{\mathrm{redundant}})$, and hence removing $R_{\mathrm{redundant}}$ does not alter the existence of a basis in $E-R_{\mathrm{redundant}}$.

Now we analyze the round complexity of \cref{alg:findBasis}. In each iteration, \cref{lem:decomposition} implies that the iterative peeling process requires at most $O(n^{1/3} / \log^{2/3}(n))$ rounds, and the stated conditions hold with probability $\geq 1 - 2^{-n/2}$. After that, in the first case, we find a set $S_k$ with $\alpha(S_k)/|S_k|\geq 1/100\log (n)$. Then by \cref{lem:decomposition}, \cref{item:findBigIS}, a random sample of
\[
\ell = \frac{\alpha(S_k)}{10|S_k|} \cdot |\M - S_1 - \dots -S_{k-1}|
\]
elements from $\M - S_1 - \dots -S_{k-1}$ is independent with probability at least $1/4$. Since the iterative peeling process terminates when fewer than half of the elements remain, we have $|\M - S_1 - \dots -S_{k-1}|\geq n/2$. Combining this with $\alpha(S_k)/|S_k| \geq 1/(100\log n)$ gives $\ell = \Omega(n/\log (n))$. Therefore, after $n^2$ (parallel) independent samples of $\ell$ elements, at least one trial produces an independent set of size $\ell = \Omega(n/\log (n))$ with probability at least $1-(1/4)^{n^2}\geq 1-2^{-n}$. Moreover, this step requires only one additional round.

In the second case, the decomposition algorithm recovers sets $S_1,\dots,S_k$ with $\sum_{i\in [k]} |S_i|\geq n/4$. We then apply the redundant elements recovery algorithm to all sets $S_1,\dots,S_k$ in parallel. For each $i\in [k]$, by \cref{thm:redundantElements}, we recover a set $R_i$ of size $\Omega(|S_i|/\log (n))$. Thus, the total number of redundant elements collected is 
\[
\sum_{i\in [k]} |R_i|=\sum_{i\in [k]} \Omega(|S_i|/\log (n)) = \Omega(n/\log (n)).
\]
Moreover, this step requires only one additional round overall, since all instances are executed in parallel, and this step succeeds with probability $\geq 1 - 2^{-n}$.

To conclude, each iteration of \cref{alg:findBasis} takes $O(n^{1/3} / \log^{2/3}(n))$ rounds, and either finds an independent set of size $\Omega(n/\log (n))$, or identifies $\Omega(n/\log (n))$ redundant elements. Let $T(n)$ denote the total number of adaptive rounds required to find a basis on any $n$-element matroid. Then we have the recurrence
\[
T(n)= T(n-\Omega(n/\log (n))) + O(n^{1/3} / \log^{2/3}(n)).
\]
Solving it yields $T(n)=O(n^{1/3}\log^{1/3} n)$. Finally, each iteration performs only polynomially many independence queries, and hence the overall query complexity is polynomial in $n$. 

Each round of the algorithm has success probability $\geq 1 - 2\cdot2^{-n/2}$. Thus, taking a union bound over the $O(n^{1/3}\log^{1/3}(n))$ many rounds, we obtain our desired theorem. 
\end{proof}

\section{Conclusion}

In this work, we present an $O(n^{1/3} \log^{1/3}n)$ round parallel algorithm for computing a matroid basis when given access only to an independence oracle. This effectively resolves an open question of Karp, Upfal, and Wigderson \cite{KUW85}. Nevertheless, there are still directions for future work.

\paragraph{Deterministic Algorithms} Our algorithm relies very heavily on randomness, contrary to the work of \cite{KUW85} which is deterministic. Note that because the number of matroids grows as $2^{2^n}$, it is not clear that one can even \emph{non-uniformly} derandomize our algorithms. We view it as an interesting open question to design deterministic algorithms for finding matroid bases with $n^{1/2 - \Omega(1)}$ many rounds of polynomially many independence queries. 

\paragraph{Work Efficiency} Our algorithms make $n^{O(1)}$ independence queries across all rounds. However, there is a simple greedy algorithm which makes $n$ queries (and uses $n$ rounds), which succeeds in finding a matroid basis. Likewise, the $O(\sqrt{n})$ round algorithm from the work of \cite{KUW85} can be made \emph{nearly}-work efficient with only a minor blow-up in the number of rounds. Formally, by using binary search, one can get an $O(\sqrt{n}\log(n))$ round algorithm for finding a matroid basis which makes a total of $O(n \log(n))$ many independence queries. Can one achieve a round complexity of $n^{1/2 - \Omega(1)}$ with only $\widetilde{O}(n)$ many queries?

\bibliographystyle{alpha}
\bibliography{ref}

\appendix

\section{Proof of \cref{thm:matroidIntersection}}

As in \cite{KPS25,KPS26}, we use the following lemma of \cite{BT25}:

\begin{lemma}[\cite{BT25}]
Let $\M_1=(E,\I_1),\M_2=(E,\I_2)$ be 2 matroids. Let $n=|E|$ and $r$ be the size of the largest independent set of $\M_1,\M_2$. Then
\begin{itemize}
    \item There is an $\widetilde O\left(\frac{nT(n)}{\varepsilon\Delta}\right)$ rounds independence-query algorithm that finds a common independent set $S\in I_1\cap I_2$ of size $|S|\geq r-(\varepsilon r+\Delta)$, given that there is a $T(n)$ round independence-query algorithm that finds a basis of a matroid on $n$ elements.
    \item Given $S\in I_1 \cap I_2$, in a single round of independence query, one can compute an $S'\in I_1 \cap I_2$ of size $|S'|=|S|+1$ or decide that $S$ is of maximum possible size.
\end{itemize}
\end{lemma}

Now, we show \cref{thm:matroidIntersection}.

\begin{proof}[Proof of \cref{thm:matroidIntersection}]
We set $\varepsilon = n^{4/9}r^{-2/3}$ and $\Delta = \varepsilon r= n^{4/9}r^{1/3}$. We can first find a common independent set $S\in I_1\cap I_2$ of size $|S|\geq r-(\varepsilon r +\Delta)$ in
\[
\widetilde O\left( \frac{n\cdot n^{1/3}}{\varepsilon \Delta} \right) = \widetilde O(n^{4/9} r^{1/3})
\]
rounds and then augment it to optimal in $O(\varepsilon r+\Delta)=O(n^{4/9}r^{1/3})$ rounds. As $r\leq n$, the total number of rounds needed is $\widetilde O(n^{7/9})$.
\end{proof}

\section{Proof of \cref{thm:MonotoneSubmodularOpt}}

The work of \cite{BRS19b} relies on the following key definition:

\begin{definition}[\cite{BRS19b}]
Given a matroid $\M=(E,\I)$, we say that $(a_1,\dots,a_{\rank(\M)})$ is a \emph{random feasible sequence} if for all $i\in [\rank(\M)]$, $a_i$ is an element chosen uniformly at random from $\{a\in E\setminus \{a_1,\dots,a_{i-1}\}:\{a_1,\dots,a_{i-1},a\}\in \I\}$.
\end{definition}

By using \cref{thm:1-3roundsintro}, we trivially obtain the following algorithm for generating feasible sequences:

\begin{lemma}\label{lem:generate-random-sequence}
Let $\M=(E,\I)$ be a matroid on $n$ elements. There is an $\tilde O(n^{1/3})$ round (polynomial independence query) algorithm that generates a random feasible sequence in $\M$.
\end{lemma}
\begin{proof}
As in \cite{KPS26}, we follow Algorithm 4 of \cite{BRS19b}. First, we sample a random permutation $\{e_1,\dots,e_n\}$ of the ground set $E$, and compute $r_i=\rank(\{e_1,\dots,e_i\})$ for every $i\in [n]$. This step takes $\tilde O(n^{1/3})$ rounds using $n$ parallel calls to our basis-finding algorithm from \cref{thm:1-3roundsintro}. Then for every $i\in [\rank(\M)]$, we let $a_i$ be the $i$th element $e_j$ such that $r_j-r_{j-1}=1$. The resulting sequence $(a_1,\dots,a_{\rank(\M)})$ is a random feasible sequence, essentially by definition,. as shown in Lemma 15 of \cite{BRS19b}.
\end{proof}

Importantly, \cite{BRS19b} shows the following blackbox reduction from submodular optimization to generating feasible sequences:

\begin{lemma}[\cite{BRS19b}]\label{lem:random-sequence}
Let $\M$ be a matroid on $n$ elements. Given that there is a $T(n)$ round independence-query algorithm that generates a random feasible sequence in $\M$, for every $\epsilon>0$, there is an $\tilde O(T(n)\cdot \epsilon ^{-3})$ round algorithm that obtains an $1-1/e-O(\epsilon)$ approximation for maximizing a monotone submodular function under the matroid constraint $\M$ with high probability.
\end{lemma}

As a consequence, we obtain \cref{thm:MonotoneSubmodularOpt}.

\section{Partition Matroids}\label{sec:partitionMatroids}

We start by defining partition matroids.

\begin{definition}[Partition Matroids] A partition matroid $\M=(E,\I)$ is defined by a ground set $E$ being partitioned into disjoint sets $A_1,\dots,A_m$ and $m$ integers $b_1,\dots,b_m$ where $0\leq b_i\leq |A_i|$. A set $S\subseteq E$ is independent iff $|S\cap A_i|\leq b_i$ for every $0\leq i\leq m$. We refer to $A_i$ as a \emph{part}, and $b_i$ as the \emph{budget} of the part.
\end{definition}

Recall the procedure $\text{RecoverSinglePart}(\M,\pi)$ from \cite{KPS25}. It takes as input a partition matroid $\M=(E,\I)$, and a permutation $\pi:[n]\to E$. The procedure runs in a single round and recovers the part of the partition matroid $\M$ that triggers the first circuit $C_{\pi}$. In particular, it returns a pair $(T,\ell)$, where $T$ is the identified part and $\ell$ is its budget. See Algorithm 1 and Claim 3.2 of \cite{KPS25} for the details.

Now we describe the basis finding algorithm utilizing this procedure. It first applies the procedure $\text{RemoveSmallParts}(\M)$ (see Algorithm 2 in \cite{KPS25}) to remove parts of budgets at most $50$. Then the algorithm proceeds in iterations. In each iteration, it uses $\text{poly}(n)$ random permutations to identify parts of the matroid. In the end of the iteration, it contracts on an independent set from the discovered parts and removes these parts from the matroid.

\begin{algorithm}[H]
\caption{FindBasisPartitionMatroids$(\M=(E,\I))$}\label{alg:findBasisPartitionMatroids}
    $\M\gets \text{RemoveSmallParts}(\M)$ \tcp{Algorithm 2 in \cite{KPS25}}
    $n \gets |E|,k \gets 0$ \\
    $B\gets \emptyset$\\
    \While{$|S_k|\leq n/4$ and $|\M|\geq n/2$}{
    $k\gets k +1$\\
    $\mathcal{A}\gets \emptyset,I\gets \emptyset$\\
    \For{$i\in [n^{15}]$} {
        Draw a random permutation $\pi$\\
        Query $\Ind(\{\pi(1),\dots,\pi(|\M|/10)\})$\\
        \If{$\Ind(\{\pi(1),\dots,\pi(|\M|/10)\})=1$} {
            $I\gets \{\pi(1),\dots,\pi(|\M|/10)\}$
        }
        $T,\ell \gets \text{RecoverSinglePart}(\M,\pi)$\\
        $\mathcal{A}\gets \mathcal{A}\cup \{(T,\ell)\}$
    }
    \If{$I\neq \emptyset$}{
        Recursively call  $\text{FindBasisPartitionMatroids}(\M/I)$ and terminate.\\ \label{line:contract}
    }
    \For{$(T,\ell)\in \mathcal{A}$} {
        $S_k\gets S_k\cup T$\\
        Pick an arbitrary $P\in {T\choose \ell}$, $I\gets I\cup P$.
    }
    $\M \leftarrow (\M/I)\setminus S_{k}$\\
    }
    Recursively call $\text{FindBasisPartitionMatroids}(\M)$. \label{line:end}
\end{algorithm}

Consider the $i$th iteration of the algorithm and let $S_i$ be the set found. The following claim, together with an union bound, implies that $q_{S_i}\geq 1-n^{-7}$ (recall that $q_{S_i} = \Pr_{\pi}[C_\pi\subseteq S_i]$).

\begin{claim} [Claim 3.5 of \cite{KPS25}]
In the $i$th iteration of \cref{alg:findBasisPartitionMatroids}, let $A_1,\dots,A_{\ell}$ be the remaining parts. For every $j\in [\ell]$, let $p_j$ denote the probability that the procedure $\mathrm{RecoverSinglePart}(\M,\pi)$, when given a random permutation $\pi$, returns $A_j$. Then with probability at least $1-2^{-n}$, we have for any $j\in[\ell]$ with $p_j\geq n^{-8}$, $A_j\subseteq S_i$.
\end{claim}

Next we show that if $\alpha(S_i)/|S_i|\geq 1/2$, the algorithm finds an independent set of size $\Omega(n)$ and terminates at \cref{line:contract}.

\begin{claim}[Claim 3.6 of \cite{KPS25}]
If $\alpha(S_i)/|S_i|\geq 1/2$, then the algorithm terminates at \cref{line:contract} in this iteration with probability $1-2^{-n}$.
\end{claim}

Let $S_1,\dots,S_k$ be the sets produced by the algorithm before termination. We obtain the same characterization of how the $\alpha$-value evolve as in \cref{clm:peelingSetsAlpha}, and the proof is identical.

\begin{claim}
Let $\mathcal{M}$ be a matroid, and let $S_1, \dots, S_k$ be a sequence of sets computed in accordance with \cref{alg:findBasisPartitionMatroids}. For any $i<j<k$ where $|S_i|\leq 2|S_j|$, we have
\[
\alpha(S_j) \geq \frac{\alpha(S_i)|S_j|}{|S_i|} + \sqrt{\frac{\alpha(S_i)|S_j|}{|S_i|} \cdot \frac{\log n}{16}}.
\]
\end{claim}

Analogous to \cref{lem:quadraticGrowth} and \cref{lem:peel-round}, we obtain the following lemmas. The proofs proceeds in the same way, except that the assumption $\alpha(S_i)/|S_i|\leq 1/100\log (n)$ there is replaced by a weaker condition $\alpha(S_i)/|S_i|\leq 1/2$ here.
\begin{lemma}
Let $\mathcal{M}$ be a matroid, and let $S_1, \dots, S_k$ be a sequence of sets computed in accordance with \cref{alg:findBasisPartitionMatroids}. Now, let $\ell\in [\log n]$ be an integer, let $T = \{ i \in [k]: |S_i| \in [2^{\ell}, 2^{\ell+1} -1]\}$, let $\gamma = |T|$, and let $a_1, \dots ,a_{\gamma}$ denote the indices in $T$. Then, 
\[
    \gamma = O\left(\sqrt{2^{\ell} /\log n}\right).
\]
\end{lemma}

\begin{lemma}\label{lem:roundPartitionMatroids}
Let $\mathcal{M}$ be a matroid, and let $S_1, \dots ,S_k$ be a sequence of sets computed in accordance with \cref{alg:findBasisPartitionMatroids}. We have $k=O(n^{1/3}/\log ^{1/3}(n))$.
\end{lemma}

Now we are ready to prove the main theorem.

\begin{theorem}
There is a randomized algorithm that, given any partition matroid $\M$ on $n$ elements, finds a basis of $\M$ with high probability in $O(n^{1/3}/\log^{1/3} (n))$ adaptive rounds, using only polynomially many independence queries per round.
\end{theorem}
\begin{proof}
It is straightforward to verify that \cref{alg:findBasisPartitionMatroids} correctly finds a basis of the matroid. For the round complexity, note that by \cref{lem:roundPartitionMatroids}, the algorithm performs at most $O(n^{1/3}/\log^{1/3}(n))$ iterations. Since each iteration uses only $O(1)$ rounds, the total number of rounds spent before recursion and termination is $O(n^{1/3}/\log^{1/3}(n))$. If the algorithm terminates at \cref{line:contract}, then it finds and contracts an independent set of size $|\M|/10\geq n/20$. On the other hand, if it terminates at \cref{line:end}, then the sets $S_1,\dots,S_k$ found have total size at least $n/4$, and they are all deleted from the matroid. Therefore, in either case, the algorithm makes $\Omega(n)$ progress towards finding a basis.

Let $T(n)$ denote the total number of adaptive rounds required to find a basis on any $n$-element partition matroid. Then we have the recurrence
\[
T(n)= T(n-\Omega(n)) + O(n^{1/3} / \log^{1/3}(n)).
\]
Solving it yields $T(n)=O(n^{1/3}/\log^{1/3} n)$. Finally, each iteration performs only polynomially many independence queries, and hence the overall query complexity is polynomial in $n$. This completes the proof.
\end{proof}

\section{Deterministic Algorithms for Linear Matroids}

First, we use the fact that the \emph{number} of linear matroids is bounded:

\begin{claim}[\cite{nelson2016almost}]\label{clm:linearMatroidBound}
    For $n \geq 12$, there are at most $2^{n^3/4}$ linear matroids on $n$ elements.
\end{claim}

\begin{proof}[Proof of \cref{thm:1-3roundsLinearIntro}]
    We consider the following basic algorithm: we run $n^4$ instances of \cref{alg:findBasis} independently on the matroid $\M$. Denote the resulting outputs of the algorithms by $\mathcal{A}_{1, \mathcal{R}}(\M),\dots \mathcal{A}_{n^4, \mathcal{R}}(\M)$, where we use  $\mathcal{R}$ to denote the random bits used by the algorithm.
    
    By \cref{thm:formalBasisFinding}, we know that each trial of \cref{alg:findBasis} succeeds in finding a basis with probability $\geq 1/2$. Thus, with probability $\geq 1 - 2^{-n^4}$ over the randomness $\mathcal{R}$, at least one of $\mathcal{A}_{1, \mathcal{R}}(\M),\dots \mathcal{A}_{n^4, \mathcal{R}}(\M)$ is a basis.

    Because there are only $\leq 2^{n^3/4}$ many linear matroids on $n$ elements (see \cref{clm:linearMatroidBound}), we then know that, 
    \[
    \Pr_{\mathcal{R}}[\forall \text{ linear matroids }\M: \text{ at least one of }\mathcal{A}_{1, \mathcal{R}}(\M),\dots \mathcal{A}_{n^4, \mathcal{R}}(\M) \text{ is a basis}] 
    \]
    \[
    \geq 1 - 2^{-n^4} \cdot 2^{n^3/4} \geq 1/2.
    \]
    This implies that there \emph{exists} a single choice of $\mathcal{R}$ such that for every linear matroid $\M$, at least one of $\mathcal{A}_{1, \mathcal{R}}(\M),\dots \mathcal{A}_{n^4, \mathcal{R}}(\M) $ is a basis of $\M$.

    Finally, observe that, conditioned on one of $\mathcal{A}_{1, \mathcal{R}}(\M),\dots \mathcal{A}_{n^4, \mathcal{R}}(\M) $ being a basis, there is a simple one round algorithm for identifying the basis. One simply queries the independent oracle with each of $\mathcal{A}_{1, \mathcal{R}}(\M),\dots \mathcal{A}_{n^4, \mathcal{R}}(\M) $, and keeps whichever output is both independent and of the largest size. This concludes the theorem. 
\end{proof}

\end{document}